\documentclass{article}
\usepackage{graphicx}
\usepackage{amsfonts}
\usepackage{amsmath}
\usepackage{amssymb}
\usepackage{amsthm}
\usepackage[margin=0.95in]{geometry}
\usepackage{tcolorbox}
\usepackage{url}
\usepackage{hyperref}
\usepackage{multirow}
\usepackage{booktabs}
\usepackage{enumitem}
\usepackage[normalem]{ulem}
\usepackage{natbib} 
\usepackage{authblk}
\usepackage{framed}

\title{Feature-based Uncertainty Model for School Choice}

\author[1]{Yao Zhang\thanks{zhang@agent.inf.kyushu-u.ac.jp}}
\author[1]{Makoto Yokoo\thanks{yokoo@inf.kyushu-u.ac.jp}}
\affil[1]{Kyushu University, Japan}

\date{}

\newtheorem{definition}{Definition}

\newtheorem{proposition}{Proposition}

\newtheorem{corollary}{Corollary}
\newtheorem{example}{Example}
\newtheorem{theorem}{Theorem}
\newtheorem{lemma}{Lemma}

\usepackage{lineno}

\begin{document}

\maketitle

\begin{abstract}
In this work, we consider a school choice scenario where a student does not exactly know which college is better for her. Although it is hard for a student to obtain an exact preference, she can usually compare specific features of colleges, such as reputation, location, and campus facilities. Motivated by this, we propose a feature-based uncertainty model for school choice where a student's preference is based on a linear combination of her utilities over different features, and the coefficients of the combination are treated as random variables. Our main goal is to achieve a higher probability of stability (ProS) and incentive compatibility (IC) for students. Unfortunately, these two goals are incompatible in general. We show that a student-proposing deferred acceptance (DA) that prioritizes colleges with higher expected ranking can achieve a worst-case approximation ratio of $(1/n)^n$ on ProS, while a DA with a carefully defined iterated comparison vector can guarantee the strongest achievable form of IC. Finally, we provide additional results for some specific restrictions on the model.
\end{abstract}

\section{Introduction}
School choice, as a classical many-to-one matching problem, has been extensively studied for decades due to its wide range of applications~\citep{gale2013college}. The primary objective is typically to design a procedure that assigns students to colleges such that the resulting matching satisfies desirable properties such as efficiency, stability, and incentive compatibility~\citep{abdulkadirouglu2003school}. Beyond the classical model, numerous extensions incorporating additional settings or constraints have been proposed to capture more realistic application scenarios—for example, constraints on college seats~\citep{haeringer2009constrained,aziz2019matching,kimura2025multi}, or networked structure among students~\citep{cho2022two,takeshima2025new}.

In most of the literature on school choice theory, it is typically assumed that students have well-defined preferences over all colleges. In reality, however, it is often difficult for students to express clear and consistent preferences due to time constraints and uncertainty about the future~\citep{bell2009college,hoxby2015high,owen2020student}. To address this, some studies introduce simplified or restricted models of student preferences~\citep{hastings2007preferences,chu2024stable}. Nevertheless, students are generally aware of which features of schools they value and can often compare schools based on specific attributes such as course suitability, academic reputation, employment prospects, or teaching quality~\citep{soutar2002students}. The main challenge for a student lies in not knowing which of these attributes will ultimately prove most important to her before experiencing them firsthand. 

Therefore, motivated by this observation, we propose a \emph{feature-based uncertainty model} for school choice, in which each student has a well-defined utility function for every feature of a college. A student’s overall preference over colleges is determined by a linear combination of these feature utilities, where the coefficients in the combination are treated as random variables. The probability distributions of these coefficients are assumed to be known to the students themselves. On the other hand, colleges are assumed to have complete and deterministic preferences over students, as in the classical model.

Under this model, our main objective is to design a matching algorithm that ensures both stability and incentive compatibility for students. Regarding stability, since students’ preferences involve uncertainty, we consider the probability of stability (ProS) and aim to find a matching that maximizes this probability. We show that, in general, this problem is computationally NP-hard, and therefore we instead seek approximation algorithms. For incentive compatibility (IC), the same uncertainty requires us to evaluate the probability of improvement that a student can obtain by misreporting her information. The strongest form of IC would guarantee that no possible improvement can ever be achieved through misreporting, which is too strong to obtain meaningful algorithms. A milder version ensures that the probability of any improvement from misreporting is less than $1/2$, though this condition is incompatible with any non-zero approximation of ProS. The weakest form only requires that certain improvements from misreporting are prevented. We examine a family of student-proposing deferred acceptance algorithms with various well-designed proposing orders to achieve different combinations of these properties, as summarized in Table~\ref{tab:results}.

The overall paper is organized as follows. Section~\ref{sec:relatedwork} briefly reviews the most related literature. Section~\ref{sec:model} introduces the model and summarizes the main results, and Sections~\ref{sec:general} to~\ref{sec:additional} present all the theoretical results. Section~\ref{sec:conclusion} includes a discussion of future work. 

\begin{table*}[t]
\caption{Summary of Theoretical Results for Four Proposing Methods: Lexicographic Order of Comparison Vectors (LOCV), Lexicographic Order of Iterated Comparison Vectors (LOICV), Higher Expected Ranking First (HERF), and Higher Expected Utility First (HEUF). IC-C corresponds to incentive compatibility (IC) that prevents certain improvement by cheating, and IC-R corresponds to IC that prevents $1/2$ probability of improvement by cheating.}\label{tab:results}
\centering
\begin{tabular}{ccc}
\toprule
\textbf{Methods} & \textbf{Worst-case Approx. of} $\mathsf{ProS}$                   & \textbf{Incentive Compatibility} \\ \midrule
LOCV                  & 0                  & IC-C    \\
\\[-.8em]
LOICV             & 0                    & IC-C (IC-R when $|F|=2$)   \\
\\[-.8em]
HERF              & $(1/n)^n$                    & IC-C   \\
\\[-.8em]
\multirow{3}{*}{HEUF}     &                   & IC-C; (when $|F|=2$, IC-R iff the probabilities of \\
& & the first feature's weight being \\
&  \multirow{-3}{*}{0} & larger or less than its expectation are equal)   \\ \bottomrule
\end{tabular}
\end{table*}
\section{Related Work}\label{sec:relatedwork}
Our work is most related to a body of literature on stable matching with uncertain preferences. \cite{aziz2020stable} first introduced the idea of stable marriage with uncertain linear preference, proposing three models of uncertainty: the lottery model, the compact indifference model and the joint probability model. They examined the computational complexity of several problems related to stability probability and showed that finding a matching with the highest stability probability is generally intractable. Later, they further introduced a model with uncertain pairwise preferences, where preferred relations can be cyclic~\citep{aziz2022stable}. \cite{alimudin2021study} also studied the stable marriage with uncertain preferences using a randomized approach. Random lottery tie-breakers have also been considered in school choice problems~\citep{abdulkadiroglu2018impact}. \cite{abdulkadi2022breaking} showed how lottery tie-breaking implements a stratified randomized trial within a DA match. In our model, if we restrict all college capacities to be one, it naturally induces a new form of uncertainty model for the stable marriage problem. Similar ideas have also been extended to other domains, including resource assignment~\citep{dvir2020modelling}, house allocation~\citep{aziz2024envy}, and committee voting~\citep{aziz2024approval}.

\subsection{Relationships to Other Uncertainty Models}
In our model, when the capacities of colleges are exactly one, then it reduces to a stable marriage problem with uncertainty. Here, we compare our uncertainty model with existing ones in the field of stable marriage. The existing models include the following ones.
\begin{itemize}
    \item Lottery Model~\cite{aziz2020stable}: for each agent, there is a probability distribution over linear preferences;
    \item Compact Indifference Model~\cite{aziz2020stable}: for each agent, a weak preference order is specified and each linear order compatible with the weak order is equally likely;
    \item Joint Probability Profiles\cite{aziz2020stable}: a probability distribution over preference profiles is specified;
    \item Pairwise Probability Model\cite{aziz2022stable}: each agent have independent pairwise comparisons on their possible partners and these comparisons my be uncertain.
\end{itemize}

For all above models, to compare with ours, we all assume the case where only one side of agents has uncertainty preferences. Then, we remark the relationships to these uncertainty models as follows. Note that here, we only consider the problems related to probability of stability.

\begin{itemize}
    \item For the lottery model, we can reduce an instance of it to our model by setting the number of features equal to the size of the union of all supports of the probability distributions over linear preferences. In this reduction, each feature’s utility function is made consistent with its corresponding linear preference, and only those weight combination characterized by unit vectors could have nonzero probabilities. Conversely, we can compute probability distributions over aggregated preferences to reduce an instance of our model to the lottery model. This indicates that both models share the same instance space.
    
    However, the reduction from our model to the lottery model is not in polynomial time, since the support of the probability distribution over the final linear preferences can be exponential even when $|F|=2$. This implies that computing problems related to the probability of stability is computationally harder in our model. Our model can thus be regarded as a more ``compact'' representation of probability distributions over linear preferences, which also constrains the space of possible distributions when $|F|$ is fixed.
    \item For the compact indifference model, as shown in the proof of Theorem~\ref{thm:nphard3}, we can always reduce any instance of it to our model. Conversely, there is no way to reduce an arbitrary instance of our model to the compact indifference model, since our model allows a wider range of possible linear preferences. Therefore, our model can be regarded as a more general framework when considering only one-sided uncertainty.
    \item The joint probability profile model additionally allows dependencies between the preferences of different agents. In contrast, our model assumes independence across agents. Therefore, while joint probability profiles may represent a more general space of instances, they also require an exponentially larger input space.
    \item In the pairwise probability model, it is assumed that pairwise preferences are independent, with probabilities specified for preferring one agent over another for every possible pair. In contrast, in our model, these pairwise probabilities are interdependent. A direct consequence of their model is that it may generate cyclic preferences among possible partners, whereas our model still ensures linear preferences. Moreover, due to the dependence of each pairwise probability, there is no direct way to reduce instances from our model to theirs while preserving the probabilities over the resulting linear preferences.
\end{itemize}
\section{The Model}\label{sec:model}
We consider a school choice scenario where a student's preference is determined by a set of measurable features, while the relative importance of these features is uncertain. We assume that every college is acceptable to each student, and vice versa. More concretely, an instance of school choice with feature-based uncertain preferences $\mathcal{I}=(N, M, X, \succ_M, F, U, \mu)$ consists of:
\begin{itemize}
    \item $N$: the set of all students, and denote them by $s_1, s_2, \dots, s_n$.
    \item $M$: the set of all colleges, and denote them by $c_1, c_2, \dots, c_m$.
    \item $X$: $m$ positive integers $x_{c_1}, x_{c_2}, \dots, x_{c_m}$, which represent the capacities of all colleges, i.e., for any $1\leq j\leq m$, college $c_j$ can accept at most $x_{c_j}$ students.
    \item $\succ_M$: the strict preferences of colleges over all students. Specifically, for a college $c\in M$, $s_i\succ_c s_j$ means $c$ prefers $s_i$ to $s_j$.
    \item $F$: the set of features. Each feature $f\in F$ represents an aspect that a student can evaluate for all colleges.
    \item $U$: the set of utility functions. Each utility function $u_s^f: M\mapsto [0,1]$ represents how much a student $s$ values each college with respect to feature $f\in F$. We assume that these evaluations are normalized to $[0,1]$. This normalization does not affect generality, since the aggregated preference of student $s$ assigns different weights to different features.
    \item $\mu$: the set of probability distributions of weights for different features. For each student $s$, $\mu_s$ is a probability distribution over all possible weight vectors, representing her uncertainty about the exact preference over colleges. Specifically, a weight vector $w_s$ is $|F|$-dimensional and satisfies $w_s^f \geq 0$ for any $f\in F$ and $\sum_{f\in F} w_s^f = 1$. An aggregated preference with $w_s$ and $\{u_s^f\}_{f\in F}$ is defined by the weighted utility\footnote{Note that we will have a tie when the equality holds.}, i.e., $c_i\succeq_s^{w_s} c_j$ if and only if $\sum_{f\in F}w_s^fu_s^f(c_i) \geq \sum_{f\in F} w_s^fu_s^f(c_j)$. We assume that $\mu_s$ and $\mu_{s'}$ are independent for any distinct students $s,s'\in N$, and both the p.d.fs ($\mathrm{Pr}[w_s =w]$) and c.d.fs ($\mathrm{Pr}[\forall f, w_s^f\leq w^f]$) can be computed in constant time.
\end{itemize}

A matching $\pi$ specifies the pairs formed between students and colleges. For simplicity of notation, let $\pi(s)$ denote the college matched to a student $s\in N$, and let $\pi(s) = \mathsf{null}$ if $s$ is unmatched; let $\pi(c)$ denote the set of students matched to a college $c\in M$, and $\pi(c) = \emptyset$ if $c$ is unmatched. A matching algorithm is a procedure that returns a matching for any given any instance of a matching problem. We consider the following properties for matching algorithms.

\noindent\textbf{Probability of Stability ($\mathsf{ProS}$).} Given a matching $\pi$ and the corresponding instance, we can examine whether $\pi$ is stable under any possible aggregated preferences of the students. Since an aggregated preference of a student may contain ties, we consider a weaker notion of stability where a block must incur a strict improvement.
\begin{definition}
    Given aggregated preferences of all students $\{\succeq_s^{w_s}\}_{s\in N}$, a matching $\pi$ is (weak) stable if there is no (strong) block $(s,c)$, $s\in N$, $c\in M$, such that $c\succ_s^{w_s} \pi(s)$, and $|\pi(c)| < x_c$ or $\exists s'\in \pi(c)$ with $s\succ_c s'$.
\end{definition}
Given an instance, the probability of stability of a matching $\pi$, denoted by $\mathsf{ProS}(\pi;\mathcal{I})$, is defined as the probability that $\pi$ is stable under aggregated preferences sampled\footnote{If there is no ambiguity, we also simplify it by omitting $\mathcal{I}$ as $\mathsf{ProS}(\pi)$}, according to the weight distributions $\mu$ in $\mathcal{I}$. For each instance, a matching $\pi^*$ is called an optimal matching if, for any other matching $\pi'$, we have $\mathsf{ProS}(\pi^*)\geq \mathsf{ProS}(\pi')$. Note that for any aggregated preference, there exists at least one stable matching, implying $\mathsf{ProS}(\pi^*)> 0$. Our goal is to find a matching algorithm that achieves a $\mathsf{ProS}$ value as close as possible to the optimal one for any instance. Therefore, for a matching algorithm, we consider the \textbf{worst-case approximation ratio} with respect to an optimal matching.
\begin{definition}
    A matching algorithm $\mathsf{Alg}$ is $\alpha$-optimal, if it satisfies $\inf_{\mathcal{I}} \frac{\mathsf{ProS}(\mathsf{Alg}(\mathcal{I});\mathcal{I})}{\max_\pi\mathsf{ProS}(\pi;\mathcal{I})}\geq \alpha$. $\mathsf{Alg}(\mathcal{I})$ denotes the output matching.
\end{definition}

\noindent\textbf{Incentive Compatibility (IC).} We consider incentive compatibility from the students' side. In traditional settings, it means that a student may strategically misreport her preferences over colleges. In our setting, a student $s$ can misreport both her utility functions $\{u_s^f\}_{f\in F}$ and probability distribution $\mu_s$. Intuitively, incentive compatibility requires that no student can obtain a better outcome by strategically misreporting her information. Note that in our model, whether an outcome is actually better may itself be uncertain for a student, since her report must be decided \textit{ex ante}. Therefore, we introduce several levels of IC, and by definition, we have that IC-A implies IC-R and IC-R implies IC-C.

\begin{definition}
    For any instance $\mathcal{I}$ and any student $s\in N$, for simplicity, let $\pi(s;u_s,\mu_s)$ denote the college assigned to $s$ by the matching algorithm $\mathsf{Alg}$ when she reports truthfully, and let $\pi'(s;u_s',\mu_s')$ denote the college matched to her under a misreport (either may be $\mathsf{null}$). Then, the algorithm $\mathsf{Alg}$ is said to be
    \begin{enumerate}   
        \item \textbf{incentive compatible with certainty (IC-C)} if $s$ cannot obtain a college that is certainly better than her truthful assignment, i.e., $\mathrm{Pr}[\pi'(s;u_s',\mu_s') \succ_s \pi(s;u_s,\mu_s)]$ $= 1$ cannot occur, for all $u_s, \mu_s, u'_s,\mu'_s$;
        \item \textbf{incentive compatible with rationality (IC-R)} if $s$ cannot obtain a college that has a higher probability of being better than her truthful assignment, i.e., $\mathrm{Pr}[\pi'(s;u_s',\mu_s') \succ_s \pi(s;u_s,\mu_s)]$ $> 1/2$ cannot occur, for all $u_s, \mu_s, u'_s,\mu'_s$;
        \item \textbf{incentive compatible with adventurism (IC-A)} if $s$ cannot obtain a college that has any chance of being better than her truthful assignment, i.e., $\mathrm{Pr}[\pi'(s;u_s',\mu_s') \succ_s \pi(s;u_s,\mu_s)]$ $ > 0$ cannot occur, for all $u_s, \mu_s, u'_s,\mu'_s$.
    \end{enumerate}
\end{definition}

\noindent\textbf{Remark.} In particular, when all the probability distributions in $\mu$ are continuous, it always holds that $\mathrm{Pr}[\pi'(s) \succeq_s \pi(s)] = \mathrm{Pr}[\pi'(s) \succ_s \pi(s)]$.
\section{Impossibilities and Intractability}\label{sec:general}
In this section, we present several impossibility and intractability results when seeking a matching with a higher probability of stability. All these negative results occur even when the number of features is $|F|=2$. We begin by stating a useful lemma that provides a formula for calculating the pairwise probability that one college is preferred over another by a student.

\begin{lemma}\label{lem:probformula}
    If the number of features equals to 2, i.e., $F = \{f_1, f_2\}$, then for any student $s\in N$ and any two colleges $c_1, c_2\in M$, denote $\Delta_s^f(c_i,c_j) = |u_s^f(c_i) - u_s^f(c_j)|$ for any $f\in F$, and $\eta_{s}(c_i,c_j) = 1\left/\left( 1+\frac{\Delta_s^{f_1}(c_i,c_j)}{\Delta_s^{f_2}(c_i,c_j)}\right)\right.$ (define $\eta_{s}(c_i,c_j)=0$ if $\Delta_s^{f_2}(c_i,c_j) = 0$). We have
    \begin{align*}
        &\mathrm{Pr}[c_i\succ_s^{w_s} c_j] = \\ &\begin{cases}
        1 & \text{if } u_s^{f_1}(c_i) > u_s^{f_1}(c_j) \text{ and } u_s^{f_2}(c_i) > u_s^{f_2}(c_j);\\
        0 & \text{if } u_s^{f_1}(c_i) \leq u_s^{f_1}(c_j) \text{ and } u_s^{f_2}(c_i) \leq u_s^{f_2}(c_j);\\
        \mathrm{Pr}[w_s^{f_1} > \eta_s(c_i,c_j)] & \text{if } u_s^{f_1}(c_i) > u_s^{f_1}(c_j) \text{ and } u_s^{f_2}(c_i) \leq u_s^{f_2}(c_j);\\
        \mathrm{Pr}[w_s^{f_1} < \eta_s(c_i,c_j)] & \text{if } u_s^{f_1}(c_i) \leq u_s^{f_1}(c_j) \text{ and } u_s^{f_2}(c_i) > u_s^{f_2}(c_j),\\
        \end{cases}
    \end{align*}
    and the case of $\mathrm{Pr}[c_i\succeq_s^{w_s} c_j]$ can be generated by $1 -\mathrm{Pr}[c_j\succ_s^{w_s} c_i]$.
\end{lemma}
\begin{proof}
    The first two cases are obvious, and the last two cases are symmetric. Hence, we only need to show the third case here. Then, we have
    \begin{align*}
        & \mathrm{Pr}[c_i\succeq_s^{w_s} c_j] \\
        & = \mathrm{Pr}[w_s^{f_1}u_s^{f_1}(c_i) + w_s^{f_2}u_s^{f_2}(c_i)\geq w_s^{f_1}u_s^{f_1}(c_j) + w_s^{f_2}u_s^{f_2}(c_j)] \\
        & = \mathrm{Pr}[w_s^{f_1}(u_s^{f_1}(c_i) - u_s^{f_1}(c_j)) \geq w_s^{f_2}(u_s^{f_2}(c_j) - u_s^{f_1}(c_i))] \\
        & = \mathrm{Pr}[w_s^{f_1}\Delta_s^{f_1}(c_i,c_j) \geq (1-w_s^{f_1})\Delta_s^{f_2}(c_i,c_j)] \\
        & = \mathrm{Pr}[w_s^{f_1}(\Delta_s^{f_1}(c_i,c_j) + \Delta_s^{f_2}(c_i,c_j)) \geq \Delta_s^{f_2}(c_i,c_j)].
    \end{align*}
    If $\Delta_s^{f_2}(c_i,c_j) = 0$, the probability equals to 1; if $\Delta_s^{f_2}(c_i,c_j) \neq 0$, the probability equals to
    \[\mathrm{Pr}\left[ w_s^{f_1} \geq \frac{\Delta_s^{f_2}(c_i,c_j)}{\Delta_s^{f_1}(c_i,c_j)+\Delta_s^{f_2}(c_i,c_j)} \right] = \mathrm{Pr}[w_s^{f_1} \geq \eta_s(c_i,c_j)].\]
    The similar equations could be derived for the cases with strict inequalities. This concludes the proof.
\end{proof}

We observe that, when $|F|=2$, the probability that one college is preferred over another is equivalent to the probability that the weight of the first feature lies within a certain interval. Based on this observation, we can derive an efficient method to calculate the probability of stability for any given matching.

\begin{proposition}\label{prop:calProS2}
    Given an instance $\mathcal{I}$ with $|F|=2$ and a corresponding matching $\pi$, we can compute the probability of stability $\mathsf{ProS}(\pi;\mathcal{I})$ in polynomial time.
\end{proposition}

\begin{proof}
    Since each probability distribution of a student's weight vector is independent from others, the probability of stability in general can be formulated as
    \[ \mathsf{ProS}(\pi;\mathcal{I}) = \prod_{s\in N} (1-\mathrm{Pr}[\exists c\in M, s.t.\ (s,c) \text{ will be a block}]). \]
    Consider each term in the product. For each student $s\in N$, define the potential set $\hat{C_s}(\pi) = \{ c\neq\pi(s) \mid |\pi(c)| < x_c \text{ or } \exists s'\in \pi(c) \text{ with } s\succ_c s', \text{ and } \mathrm{Pr}[c\succ_s^{w_s} \pi(s)] >0 \}$. Then, only colleges in the potential set may have a positive probability to make a block with $s$. If $\hat{C_s}(\pi) = \emptyset$, then the term for $s$ equals to 1. If there exists a college $c\in \hat{C_s}(\pi)$, such that $\mathrm{Pr}[c\succ_s^{w_s} \pi(s)] = 1$, then the term equals to 0 (and the whole $\mathsf{ProS}$ of $\pi$ is 0). For the remaining cases, according to Lemma~\ref{lem:probformula}, the colleges in $\hat{C_s}(\pi)$ must belong to one of the following two categories:
    \begin{itemize}
        \item those $c$ such that $u_s^{f_1}(c) > u_s^{f_1}(\pi(s))$ and $u_s^{f_2}(c) \leq u_s^{f_2}(\pi(s))$. Denote the set of this category as $\hat{C}_s^1(\pi)$. For any $c\in \hat{C}_s^1(\pi)$, $(s,c)$ will be a block when $w_s^{f_1} > \eta_s(c,\pi(s))$. Denote $\eta_s^1(\pi) = \min_{c\in \hat{C}_s^1(\pi)} \eta_s(c,\pi(s))$, and let $\eta_s^1(\pi) = 1$ if $\hat{C}_s^1(\pi) = \emptyset$. Then, $\mathrm{Pr}[\exists c\in \hat{C}_s^1(\pi), s.t.\ (s,c) \text{ is a block}] = \mathrm{Pr}[w_s^{f_1} > \eta_s^1(\pi)]$;
        \item those $c$ such that $u_s^{f_1}(c) \leq u_s^{f_1}(\pi(s))$ and $u_s^{f_2}(c) > u_s^{f_2}(\pi(s))$. Similarly, denote the set of this category as $\hat{C}_s^2(\pi)$. For any $c\in \hat{C}_s^2(\pi)$, $(s,c)$ will be a block when $w_s^{f_1} < \eta_s(c,\pi(s))$. Denote $\eta_s^2(\pi) = \max_{c\in \hat{C}_s^2(\pi)} \eta_s(c,\pi(s))$, and let $\eta_s^2(\pi) = 0$ if $\hat{C}_s^2(\pi) = \emptyset$. Then,

        $\mathrm{Pr}[\exists c\in \hat{C}_s^2(\pi), s.t.\ (s,c) \text{ is a block}] = \mathrm{Pr}[w_s^{f_1} < \eta_s^2(\pi)]$.
    \end{itemize}
    Hence, the probability of whether there exists a block associated with $s$ is the probability of whether the weight of the first feature $w_s^{f_1}$ lies in a specific interval, i.e.,
    \begin{align*}
        &1-\mathrm{Pr}[\exists c\in M, s.t.\ (s,c) \text{ will be a block}] = \\ &\begin{cases}
        0 & \text{if } \eta_s^1(\pi) < \eta_s^2(\pi);\\
        \mathrm{Pr}[\eta_s^2(\pi) \leq w_s^{f_1} \leq \eta_s^1(\pi)] & \text{if } \eta_s^1(\pi) \geq \eta_s^2(\pi),\\
        \end{cases}
    \end{align*}
    The probability in the second case can be easily achieved by an oracle that can return the value of $\mu_s$'s c.d.f. Deriving the potential set, checking each college in it and computing the values of $\eta_s^1(\pi)$ and $\eta_s^2(\pi)$ totally need $O(m)$ time. Therefore, we can finally compute $\mathsf{ProS}(\pi)$ in $O(mn)$ time, which is polynomial.
\end{proof}

However, finding a matching with the highest probability of stability is NP-hard.

\begin{theorem}\label{thm:nphard}
    It is NP-hard to find a matching with the highest $\mathsf{ProS}$ in a school choice with feature-based uncertainty, even with the constraint of $|F| = 2$.
\end{theorem}
\begin{proof}
    This proof shares the same framework of the proof of Theorem 9 in~\citep{aziz2020stable}, with some critical difference in detailed construction, so we still illustrate the whole process of the proof here. The basic idea is to reduce instances of stable marriage problem with ties and incomplete lists (SMTI)~\citep{manlove2002hard}. First, we claim that involving incomplete preferences will not affect the time complexity of our problem. It can be done by simply completing the preferences with nonexistent candidates to the end of the preferences (or the least utilities) according to a predetermined order. Then, the unmatched agents can match each other according to this order without changing the final $\mathsf{ProS}$. According to the key fact that there is no polynomial-time algorithm can distinguish whether the maximum size of a weakly stable matching in an instance of SMTI with size $n$ is larger than $\frac{41}{54}n$ or smaller than $\frac{39}{54}n$ unless P=NP, even if the ties occur on one side only, each tie is of length 2, and each preference list contains at most one tie~\citep{halldorsson2007improved}, now we can reduce such an instance to our problem as follows.
    \begin{itemize}
        \item Let the side with uncertainty corresponds to students, and the other side as the colleges. Each college has capacity of 1.
        \item Let $s_1$, $\dots$, $s_n$ be corresponding students from the SMTI instance, and add additional students $z_1$, $\dots$, $z_l$, $l = \lfloor\frac{n}{27}\rfloor$. The probability distribution of weight vector of each student is a uniform distribution over all possible weights. We set their utility functions on two features in the fourth step.
        \item For the college side, let $c_1$, $\dots$, $c_n$ be corresponding colleges from the SMTI instance, and they share the same preferences with the original instance. We additionally add following colleges. $x_1$, $\dots$, $x_k$, with $k = \lfloor\frac{13}{54}n\rfloor$, all have the preference as $s_1\succ s_2\succ \cdots \succ s_n$; $y_1$, $\dots$, $y_l$, and for each $1\leq i\leq l$, $y_i$ has the preference as $s_1\succ s_2\succ \cdots \succ s_n\succ z_i$; $y_1^j$, $\dots$, $y_t^j$, with any $1\leq j\leq l$, and $t = n^{27}$, where for each $y_i^j$, $1\leq i\leq t$, it has the preference that only contains $z_j$.
        \item For the student side, we define their utility functions as follows. For each student $s_i$, $1\leq i\leq n$, let the sequence of their corresponding preference in the SMTI instance be $c_{(1)}$, $\dots$, $c_{(m)}$ with some $m\leq n$. If there is no tie in the sequence, then for both $f\in F$, we set $u_{s_i}^f(c_{(j)}) = 1- \frac{j-1}{2(m-1)}$ for $1\leq j\leq m$, $u_{s_i}^f(x_j) = \frac{1}{2} - \frac{j}{4k}$ for $1\leq j\leq k$, and $u_{s_i}^f(y_j) = \frac{1}{4} - \frac{j}{8l}$ for $1\leq y\leq l$. If there is a tie $(c_{(j)}, c_{(j+1)})$ in the sequence, then for $f_1$, we set the same utilities as above, and for $f_2$, set $u_{s_i}^{f_2}(c_{(j)}) = u_{s_i}^{f_1}(c_{(j+1)})$, $u_{s_i}^{f_2}(c_{(j+1)}) = u_{s_i}^{f_1}(c_{(j)})$, and all other colleges have the same utilities as in $f_1$.

        For each student $z_i$, $1\leq i\leq l$, set $u_{z_i}^{f_1}(y_i) = u_{z_i}^{f_2}(y_i) = 1$, $u_{z_i}^{f_1}(y_j^i) = \nu_j$, and $u_{z_i}^{f_2}(y_j^i) = \nu_{t-j+1}$, for $1\leq j\leq t$, where the values of $\nu$ is defined as (i) $\nu_1 = 1/t$, and (ii) $\nu_j = \nu_{j-1} + (t-j+1)\delta$ for $1<j\leq t$, with $0<\delta<\frac{2}{t^2}$. Then, we have $\mathrm{Pr}\left[\bigcap_{j'\neq j} y_j^i\succeq_{s_i}^{w_{s_i}} y_{j'}^i\right] = \frac{1}{t}$ for any $1\leq j\leq t$.
    \end{itemize}
    With the above reduction, if in the original SMTI instance, the maximum size of a weakly stable matching is larger than $\frac{41}{54}n$, then at most $\lfloor \frac{13}{54}n \rfloor$ agents are unmatched on one side. We can construct an extended matching $\pi$ in the reduced instance as follows: first, apply the same matching between $s_1$, $\dots$, $s_n$ and $c_1$, $\dots$, $c_n$; next, for unmatched students among $s_1$ to $s_n$, pair them with colleges $x_1$ to $x_k$ with mutually increasing order of their indices, and all unmatched students among $s_1$ to $s_n$ will be matched since $k = \lfloor \frac{13}{54}n \rfloor$; finally, we can complete the matching with assigning each $z_i$ to $y_i$ for $1\leq i\leq l$. For this extended matching, only students among $s_1$, $\dots$, $s_n$ may be involved in potential blocks. Since each student $s_i$ has at most on uncertainty of two colleges in the aggregated preference with equal probability, the final $\mathsf{ProS}(\pi)\geq \left( \frac{1}{2} \right)^n$.

    If in the original SMTI instance, the maximum size of a weakly stable matching is less than $\frac{39}{54}n$, then at least $\lceil \frac{15}{54}n \rceil$ agents are unmatched on one side. Note that in the reduced instance, if a matching between $s_1$ to $s_n$ and $c_1$ to $c_n$ is not weakly stable in the original SMTI instance, then the final $\mathsf{ProS}$ must be 0. Hence, for a matching $\pi$ such that $\mathsf{ProS}(\pi) > 0$, there are at least $\lceil \frac{15}{54}n \rceil$ students among $s_1$, $\dots$, $s_n$ cannot be matched to a college among $c_1$, $\dots$, $c_n$. Then, all colleges of $x_1$ to $x_k$ and $y_1$ to $y_l$ should enroll a student among $s_1$, $\dots$, $s_n$ in $\pi$. Finally, each student $z_i$ can only be assigned to one of $y_j^i$, $1\leq j\leq t$, and whichever college she is assigned to, there is at most $\frac{1}{t}$ probability for her not being involved a block. Hence, the final $\mathsf{ProS}(\pi)\leq \left( \frac{1}{t} \right)^l \leq \left( \frac{1}{n} \right)^n$.

    On the other hand, if in the reduction instance, the maximal probability of stability is larger than $\left( \frac{1}{2} \right)^n$, then we must have a weakly stable matching between $s_1$, $\dots$, $s_n$ and $c_1$, $\dots$, $c_n$ with size larger than $\frac{41}{51}n$; otherwise, at least $\frac{14}{51}n$ students being unmatched with one of $c_1$, $\dots$, $c_n$ cannot lead to a $\mathsf{ProS}$ larger than $\left( \frac{1}{2} \right)^n$ according to the above analysis. Similarly, the maximal probability of stability being smaller than $\left( \frac{1}{n} \right)^n$ also leads to the fact that the maximum size of a weakly stable matching of the original SMTI instance is less than $\frac{39}{54}n$.

    Therefore, it is NP-hard to distinguish whether the largest $\mathsf{ProS}$ in our model with $|F|=2$ is larger than $\left( \frac{1}{2} \right)^n$ or less than $\left( \frac{1}{n} \right)^n$. Since calculating the probability of stability could be done in polynomial time when $|F|=2$ according to Proposition~\ref{prop:calProS2}, we can conclude that finding a matching with the largest $\mathsf{ProS}$ is NP-hard.
\end{proof}

Hence, we aim to find an algorithm that achieves a good approximation of the optimal $\mathsf{ProS}$ in general. Since we also require incentive compatibility for students, we examine what can be achieved under different levels of IC. The first impossibility result is the unattainability of a meaningful IC-A algorithm.

\begin{proposition}\label{prop:noICA}
    There is no matching algorithm that can satisfy the property of IC-A unless it always outputs the same result.
\end{proposition}
\begin{proof}
    Suppose a matching algorithm is IC-A. Then, consider a student $s$ who has utility function as $u_s^{f_1}(c_i) = i/m$, $u_s^{f_2} = (m-i+1)/m$, and $u_s^f(c_i) = a$, for any $f\neq f_1,f_2$, where $m$ is the number of colleges and $a$ is a constant between $0$ and $1$. The probability distribution $\mu_s$ is a uniform distribution over all possible weight vectors. Now, the aggregated preference of $s$ is actually only depends on $f_1$ and $f_2$, and by Lemma~\ref{lem:probformula}, we have $\mathrm{Pr}[c_i\succ_s^{w_s}c_j] = 1/2>0$ for any $c_i\neq c_j$. Suppose the algorithm in this case matches $s$ to a college $c$ ($c$ could be $\mathsf{null}$ if she is not assigned to any college). Then, for any other possible utility functions $u_s'$ and probability distribution $\mu_s'$, the algorithm must still assign $s$ to the same $s$; otherwise, $s$ can misreport her information between two cases to obtain another college that has non-zero chance to be better. Therefore, the algorithm could only output the same result for any input.
\end{proof}

Since IC-A is too strong to yield meaningful algorithms, we consider milder versions of IC. However, even IC-R is not compatible with achieving a non-zero approximation of $\mathsf{ProS}$.

\begin{theorem}
    There is no matching algorithm that can be IC-R and $\alpha$-optimal with $\alpha > 0$ even when the number of features $|F|=2$.
\end{theorem}
\begin{proof}
    We can prove the statement by giving an instance with three students and three colleges. The utility functions of three students are listed as follows.
    \begin{itemize}
        \item $s_1$: $u_{s_1}^{f_1}(c_1) = 1.5\delta + 3\epsilon,\ u_{s_1}^{f_1}(c_2) = \delta + 2\epsilon,\ u_{s_1}^{f_1}(c_3) = 0$;

        $\quad\ u_{s_1}^{f_2}(c_1) = 0,\ u_{s_1}^{f_2}(c_2) = 0.5\delta + 2\epsilon,\ u_{s_1}^{f_2}(c_3) = 1.5\delta + 2\epsilon$;
        \item $s_2$: $u_{s_2}^{f_{1/2}}(c_1) = 0.3,\ u_{s_2}^{f_{1/2}}(c_2) = 0.6,\ u_{s_2}^{f_{1/2}}(c_3) = 0.1$;
        \item $s_3$: $u_{s_3}^{f_{1/2}}(c_1) = 0.3,\ u_{s_3}^{f_{1/2}}(c_2) = 0.1,\ u_{s_3}^{f_{1/2}}(c_3) = 0.6$.
    \end{itemize}
    Here, $\delta$ and $\epsilon$ are any positive values, and $\epsilon\rightarrow 0$. We assume that the probability distributions of weights for all three students are uniform distributions over all possible weights. We can notice that, for students $s_2$ and $s_3$, whichever the weight vector is chosen, their aggregated preferences on colleges are always $c_2\succ_{s_2} c_1\succ_{s_2} c_3$ and $c_3\succ_{s_3} c_1\succ_{s_3} c_2$. On the other hand, the preferences of three colleges are $s_1\succ_{c_1/c_2} s_2\succ_{c_1/c_2} s_3$, and $s_1\succ_{c_3} s_3\succ_{c_3} s_2$.
    Here, we suppose the capacities of all colleges are one. First, we check what will happen if a matching algorithm matches $s_1$ to $c_2$. In this case, for the remaining students and colleges, we can only match $s_2$ to $c_1$ and $s_3$ to $c_3$; otherwise $(s_2,c_1)$ will be a block. Since both $c_1$ and $c_3$ ranks $s_1$ the first in their preferences, hence, the $\mathsf{ProS}$ of this matching will be
    \[ \mathrm{Pr}[c_2\succeq_{s_1}^{w_{s_1}}c_1 \text{ and } c_2\succeq_{s_1}^{w_{s_1}}c_3] = \frac{\epsilon(\delta+2\epsilon)}{(\delta+\epsilon)(\delta + 3\epsilon)}, \]
    which approaches 0 when $\epsilon\rightarrow 0$. Nevertheless, if we consider an alternative matching, where $s_1$ is matched to $c_1$, $s_2$ is matched to $c$, and $s_3$ is matched to $c$, then the $\mathsf{ProS}$ of this new matching will be
    \[ \mathrm{Pr}[c_1\succeq_{s_1}^{w_{s_1}}c_2 \text{ and } c_1\succeq_{s_1}^{w_{s_1}}c_3] = \frac{\delta+2\epsilon}{2\delta+6\epsilon}. \]
    We can observe that the ratio of these two probabilities approaches 0 when $\epsilon\rightarrow 0$. Therefore, a matching algorithm that has a non-zero approximation ratio cannot match $s_1$ to $c_2$. However, for $s_1$, $c_2$ is more likely to be better than $c_1$ or $c_3$:
    \[ \mathrm{Pr}[c_2\succeq_{s_1}^{w_{s_1}}c_1] = \frac{\delta+4\epsilon}{2\delta+6\epsilon} > \frac{1}{2};\quad \mathrm{Pr}[c_2\succeq_{s_1}^{w_{s_1}}c_3] = \frac{\delta+2\epsilon}{2\delta+2\epsilon} > \frac{1}{2}. \]
    Then, whichever college the matching algorithm matches to $s_1$, she can construct a misreporting as follows.
    \begin{itemize}
        \item $s_1$: $u_{s_1}^{f_{1/2}}(c_1) = 0.3,\ u_{s_1}^{f_{1/2}}(c_2) = 0.6,\ u_{s_1}^{f_{1/2}}(c_3) = 0.1$,
    \end{itemize}
    where it is always $c_2\succ_{s_1} c_1\succ_{s_1} c_3$. In this new case, $s_1$ must be matched to $c_2$; otherwise $(s_1, c_2)$ will be a block. Namely, $s_1$ can force a choice that will be better with probability larger than $1/2$, which violates the property of IC-R. In summary, in the given instance, $c_1$ being matched to $c_2$ violates IC-R, while $c_1$ not being matched to $c_2$ violates the non-zero approximation ratio, which infers that a matching algorithm cannot be IC-R and $\alpha$-optimal with $\alpha > 0$ at the same time. Since the instance only has two features, this impossibility holds even for the restriction of $|F| = 2$. 
\end{proof}

Hence, we can only require IC-C if we want an algorithm with a positive approximation ratio. The final result in this section provides an upper bound for the ratio in this case.

\begin{theorem}
    There is no matching algorithm that can be IC-C and $\alpha$-optimal with $\alpha>(\sqrt[3]{\phi})^{n}$, $\phi = \frac{\sqrt{5}-1}{2}$, even when $|F| = 2$. 
\end{theorem}

\begin{proof}
    Since the approximation ratio is defined in the worst case, to establish this upper bound, it suffices to show that the ratio cannot be achieved in at least one instance. Consider the following case with $n = 3k$ students and $m = n$ colleges. Each college $c_j\in M$ has capacity $x_j = 1$. For each $1\leq i\leq k$, we suppose that $s_{i+1}\succ_{c_i} s_i\succ_{c_i} s_{i+2} \succ_{c_{i}} \cdots$, $s_i\succ_{c_{i+1}} s_{i+1}\succ_{c_{i+1}} s_{i+2} \succ_{c_{i+1}} \cdots$, and $s_i\succ_{c_{i+2}} s_{i+2}\succ_{c_{i+2}} \cdots$,
    where ``$\cdots$'' means the preferences over all other students could be arbitrary. On the other hand, for the student side, we suppose the probability distributions of weights are uniform distributions, and
    \begin{itemize}
        \item $s_i$:
        $u_{s_{i}}^{f_{1}}(c_i) = 1 > u_{s_i}^{f_1}(c_{i+1}) = 0.8 > u_{s_{i}}^{f_1}(c_{i+2}) = 0.8 - y > \cdots$,
        
        $\quad\ u_{s_{i}}^{f_{2}}(c_i) = 1 > u_{s_i}^{f_2}(c_{i+2}) = 0.8 > u_{s_{i}}^{f_2}(c_{i+1}) = y - 0.2 > \cdots$;
        \item $s_{i+1}$:
        $u_{s_{i+1}}^{f_{1}}(c_{i+1}) = 1 > u_{s_{i+1}}^{f_1}(c_{i}) = 1 - z  > \cdots$,
        
        $\quad\quad u_{s_{i+1}}^{f_{2}}(c_{i}) = 1 > u_{s_{i+1}}^{f_2}(c_{i+1}) = z  > \cdots$;
        \item $s_{i+2}$: $u_{s_{i+2}}^{f_{1}}(c_{i+2}) = 1 > \cdots$, $u_{s_{i+2}}^{f_{2}}(c_{i+2}) = 1 > \cdots$,
    \end{itemize}
    where ``$\cdots$'' means the utilities of all other colleges are arbitrary but with the same order in both $f_1$ and $f_2$. Then, the only uncertainties in preferences of these three students are $\mathrm{Pr}[c_{i+1}\succ_{s_i}^{w_{s_i}} c_{i+2}] = y$ and $\mathrm{Pr}[c_{i+1}\succ_{s_{i+1}}^{w_{s_{i+1}}} c_{i}] = z$.

    Now, to have a matching with positive $\mathsf{ProS}$, students $s_{i}$, $s_{i+1}$ and $s_{i+2}$ must be matched to colleges $c_{i}$, $c_{i+1}$ and $c_{i+2}$: if $s_i$ has not been matched to one of them, $(s_{i}, c_{i+1})$ would be a block with probability one; if $s_{i+1}$ has not been matched to one of them, $(s_{i+1}, c_{i})$ would be a block with probability one; if $s_{i+2}$ has not been matched to one of them, then at least one of three colleges could not enroll both $s_{i}$ and $s_{i+1}$, which could make a block with $s_{i+2}$ of probability one. Then, for all possible matching results among these six agents, there are only three of them that could have a positive $\mathsf{ProS}$:\\
    (a) $\{(s_{i}, c_{i}), (s_{i+1}, c_{i+1}), (s_{i+2}, c_{i+2})\}$ with $\mathsf{ProS} = z$, \\
    (b) $\{(s_{i}, c_{i+1}), (s_{i+1}, c_{i}), (s_{i+2}, c_{i+2})\}$ with $\mathsf{ProS} = y$, and \\
    (c) $\{(s_{i}, c_{i+2}), (s_{i+1}, c_{i}), (s_{i+2}, c_{i+1})\}$ with $\mathsf{ProS} = (1-z)(1-y)$. \\
    Consider the two extreme cases where $y=0$ (Case 0) and $y=1$ (Case 1). In Case 0, we can only choose matching (a) or (c) to ensure a positive $\mathsf{ProS}$, while in Case 1, we can only choose matching (a) or (b). Notice that the student $s_i$ can misreport her utilities to freely turn one case into the other. If in Case 0, $s_i$ is matched to $c_i$ but in Case 1, she is not matched to $c_i$, then an $s_i$ with true type being Case 1 can misreport to pretend to be the one in Case 0, where she can be matched to her certainly best college $c_i$. Similarly, we cannot match $s_i$ to $c_i$ in Case 1, without matching her to $c_i$ in Case 0. Hence, an IC-C matching algorithm with any positive approximation ratio can only have two choices: choosing matching (a) in both cases, or choosing (c) in Case 0 and choosing (b) in Case 1. For the former one, the approximation ratio of the best $\mathsf{ProS}$ among these agents is $\min\{z, \mathbb{I}(z\geq 1/2) + \mathbb{I}(z<1/2)\cdot (z/(1-z)) \} = z$, where $\mathbb{I}(\cdot)$ is the indicator function; for the latter one, the approximation ratio of the best $\mathsf{ProS}$ among these agents is $\mathbb{I}(z\leq 1/2) + \mathbb{I}(z> 1/2)\cdot ((1-z)/z)$. Therefore, for any satisfiable matching algorithm, it will not exceed
    \[ \min_z \max\left\{ z, \mathbb{I}(z\leq 1/2) + \mathbb{I}(z> 1/2)\cdot \frac{1-z}{z} \right\} = \frac{\sqrt{5}-1}{2}. \]
    Finally, since the above result is independent of $i$, an upper bound of the approximation ratio for the whole instance will be $\phi^k = \phi^{n/3}$, where $\phi = \frac{\sqrt{5}-1}{2}$.
\end{proof}
\section{Methods}\label{sec:methods}
In this section, we propose the methods that will be examined. Since we consider both IC and stability, we focus on a generalized version of the student-proposing deferred acceptance (GDA) algorithm~\citep{roth2008deferred}, such that when the scenario degenerates to the case with no uncertainty among students, the result naturally coincides with the standard student-proposing DA algorithm.

\begin{framed}
    \noindent \textbf{Generalized Student-Proposing Deferred Acceptance}

    \noindent \rule{\textwidth}{0.5pt}
    
    \noindent \textsc{Input:} an instance $(N, M, X, \succ_M, F, U, \mu)$.
    \begin{enumerate}
    \item Initialize the matching $\pi$ to be empty.
    \item Initialize the set of unmatched students as $N_{\mathcal{U}} \leftarrow N$.
    \item For each student $s\in N$, initialize the set of colleges which have rejected her as $R_s \leftarrow \emptyset$.
    \item \textbf{while} $|N_{\mathcal{U}}|>0$ \textbf{and} for any $s\in N_{\mathcal{U}}$, $R_s\neq M$ \textbf{do}
    \begin{enumerate}
        \item Initialize $P_c$ to be empty for all $c\in M$
        \item For each student $s\in N_{\mathcal{U}}$ with $R_s\neq M$, let $s$ propose to a college $c = \mathsf{Next}(s, R_s, \{u_s^f\}_{f\in F}, \mu_s)$. Add $s$ to $P_c$.
        \item For each college $c$ such that $P_c\neq \emptyset$,\\ \textbf{if} $|\pi(c)|+|P_c|\leq x_c$ \textbf{then}
        \begin{enumerate}
            \item[] Add all students in $P_c$ to $\pi(c)$.
            \item[] For each student $s\in P_c$, update $\pi(s) = c$ and remove $s$ from $N_{\mathcal{U}}$.
        \end{enumerate}
        \textbf{else}
        \begin{enumerate}
            \item[] Choose the most preferred $x_c$ students of $c$ from $\pi_c\cup P_c$, update $\pi(c)$ as the set of these $x_c$ students and reject others.
            \item[] For each student $s$ in updated $\pi(c)$, update $\pi(s) = c$ and remove $s$ from $N_{\mathcal{U}}$ (if $s\in N_{\mathcal{U}}$).
            \item[] For each student $s$ being rejected, update $\pi(s)= \mathsf{null}$ and add $s$ to $N_{\mathcal{U}}$ (if $s\notin N_{\mathcal{U}}$).
        \end{enumerate}
        \textbf{end if}
    \end{enumerate}
    \textbf{end while}
    \end{enumerate}
    \noindent \textsc{Output:} the final matching $\pi$.
\end{framed}

For each student $s$, we can define an ordering of colleges based on the $\mathsf{Next}()$ function. Denote this ordering as $\succ_s^{\mathsf{Next}}$. Then, the output of the algorithm must be student-optimal with respect to this evaluation. Specifically, among all possible matching results in which no student $s$ can find a college that is more preferred under $\succ_s^{\mathsf{Next}}$ than her current assignment and that also prefers her to at least one of its current enrollees, each student receives her best possible option according to $\succ_s^{\mathsf{Next}}$. The focus on such student-optimal matching algorithms (e.g., those that are student-optimal in terms of expected utilities) is another reason we restrict attention to this class of methods.

Within the class of GDA, the key problem is designing a suitable $\mathsf{Next}()$ function. One of the most straightforward ways is described above: it selects the college with the highest expected utility among those not in $R_s$. Besides, recall that our main goal is seeking high $\mathsf{ProS}$ and IC, which may not be fully captured by expected utilities. Therefore, we propose three additional methods based on the probabilities of one college being preferred over another. More concretely, the four methods are defined as follows. 

\noindent\textbf{Higher Expected Utility First (HEUF).} In the method of HEUF, the $\mathsf{Next}()$ function just chooses a college such that
\[ \mathsf{Next}(s, R_s, \{u_s^f\}_{f\in F}, \mu_s)\in \arg\max_{c\notin R_s} \mathbb{E}_{w_s\sim \mu_s}\left[ \sum_{f\in F} w_s^f\cdot u_s^f(c)  \right], \]
with random or any predetermined tie-breaking. The time complexity of computing the expected utility depends on the forms of probability distributions. In real applications, it is usually a discrete distribution with finite support or a common parameterized continuous distribution, which can be computed efficiently.

\noindent\textbf{Lexicographic Order of Comparison Vectors (LOCV).} In the method of LOCV, the $\mathsf{Next}()$ function chooses a college based on a specific order of the colleges. For each student $s$, such an order is defined by a lexicographic order of \emph{comparison vectors}. Given a college $c$, the comparison vector for the student $s$ is defined as
\[ q_s(c) = \mathsf{Ascending}((\mathrm{Pr}[c\succeq_s^{w_s}c'])_{c\in M, c'\neq c}), \]
where $\mathsf{Ascending}(\cdot)$ rearranges a vector's elements in an ascending order. Then, for any two colleges $c$ and $c'$, we say $c\succ_s^{\mathsf{lex}}c'$, if there exists an integer $0\leq k< m$, the first $k$ elements in $q_s(c)$ and $q_s(c')$ is the same, while the $(k+1)$th element of $q_s(c)$ is larger than that of $q_s(c')$. Based on $\succ_s^{\mathsf{lex}}$, we can have a complete order of all colleges with random or any predetermined tie-breaking, and each time the $\mathsf{Next}()$ function returns the first college that has not rejected the student in $\succ_s^{\mathsf{lex}}$. The intuition behind LOCV is that we try to choose a college that has a higher chance of being better than another college in the worst cases. By Lemma~\ref{lem:probformula}, these probabilities needed in the method can be efficiently computed by an oracle of c.d.f when $|F| = 2$. For $|F|\geq 3$, we discuss the issue in Section~\ref{sec:moretc}.

\noindent\textbf{Lexicographic Order of Iterated Comparison Vectors (LOICV).} The only difference compared to the LOCV method is that we update the comparison vectors iteratively based on the current $R_s$ set:
\[ \tilde{q}_s(c, R_s) = \mathsf{Ascending}((\mathrm{Pr}[c\succeq_s^{w_s}c'])_{c\in M\setminus R_s, c'\neq c}). \]
Then, each time, we can have an updated complete order for remaining colleges that have not been proposed by $s$, and the $\mathsf{Next}()$ function returns the first college in such an order. The intuition behind LOICV is that we do not need to worry about colleges that have rejected $s$ making a block with her. The time complexity of the LOICV method is quadratic in that of the LOCV method, due to its additional computed probabilities.

\noindent\textbf{Higher Expected Ranking First (HERF).} In the method of HERF, we also consider the order of colleges iteratively as in LOICV. The idea is motivated by the observation that $\mathrm{Pr}[c_i\succ_s^{w_s}c_j]\geq \mathrm{Pr}[c_j\succ_s^{w_s}c_i]$ for any $j\neq i$ still does not imply $c_i$ has the highest probability to be most preferred by the student $s$ (see the third instance in Example~\ref{ex:running}). Hence, in HERF, the $\mathsf{Next}()$ function instead chooses a college such that
\[ \mathsf{Next}(s, R_s, \{u_s^f\}_{f\in F}, \mu_s)\in \arg\max_{c\notin R_s} \mathrm{Pr}\left[ \bigcap_{c'\neq c, c\notin R_s} c\succeq_s^{w_s}c'  \right], \]
with random or any predetermined tie-breaking. By Lemma~\ref{lem:probformula}, the above probability is equal to the probability that the weight of the first feature locates in an intersection of $m-1$ intervals when $|F| = 2$, which can be computed efficiently by an oracle of c.d.f. For $|F|\geq 3$, we discuss the issue in Section~\ref{sec:moretc}.

With the above methods, we can obtain student-optimal matching results with respect to expected utilities, pairwise preferences, or expected rankings, and they all degenerate to standard student-proposing DA when uncertainty disappears. In the following example, we will see that, for the objective of $\mathsf{ProS}$, none of these methods consistently outperforms the others across all instances. This suggests that each method may have its own suitable scenarios in practical applications.

\begin{example}\label{ex:running}
    We will show three instances in this example. All of them have $n=m=3$, $|F|=2$, each college $c$ has capacity $x_c = 1$, and each probability distribution $\mu_s$ is a uniform distribution. For the first instance, the preferences of colleges are $s_1\succ_{c_1} s_2 \succ_{c_1} s_3$, $s_1\succ_{c_2} s_3 \succ_{c_2} s_2$, and $s_2\succ_{c_3} s_3 \succ_{c_3} s_1$.
    On the other hand, the utility functions of students are
    \begin{itemize}
        \item $s_1$: $u_{s_1}^{f_1}(c_1) = 0.3,\ u_{s_1}^{f_1}(c_2) = 0.2,\ u_{s_1}^{f_1}(c_3) = 1.0$;

        $\quad\ u_{s_1}^{f_2}(c_1) = 0.7,\ u_{s_1}^{f_2}(c_2) = 0.4,\ u_{s_1}^{f_2}(c_3) = 0.3$;
        \item $s_2$: $u_{s_2}^{f_1}(c_1) = 0.5,\ u_{s_2}^{f_1}(c_2) = 0.1,\ u_{s_2}^{f_1}(c_3) = 0.7$;

        $\quad\ u_{s_2}^{f_2}(c_1) = 0.6,\ u_{s_2}^{f_2}(c_2) = 0.3,\ u_{s_2}^{f_2}(c_3) = 0.1$;
        \item $s_3$: $u_{s_3}^{f_1}(c_1) = 0.9,\ u_{s_3}^{f_1}(c_2) = 0.3,\ u_{s_3}^{f_1}(c_3) = 0.6$;

        $\quad\ u_{s_3}^{f_2}(c_1) = 0.2,\ u_{s_3}^{f_2}(c_2) = 0.3,\ u_{s_3}^{f_2}(c_3) = 0.1$.
    \end{itemize}
    Taking student $s_3$ as an example, the comparison vector $q_{s_3}(c_2) = (1/7, 2/5)$ since we can compute that $\mathrm{Pr}[c_2\succeq_{s_3}^{w_{s_3}}c_1] = 1/7 < \mathrm{Pr}[c_2\succeq_{s_3}^{w_{s_3}}c_3] = 2/5$. Similarly, we have $q_{s_3}(c_1) = (6/7, 1)$ and $q_{s_3}(c_3) = (0, 3/5)$. Hence, $c_1\succ_{s_3}^{\mathsf{lex}} c_2\succ_{s_3}^{\mathsf{lex}} c_3$. If applying the LOCV method, $s_2$ and $s_3$ will apply to $c_1$, and $s_1$ will apply to $c_3$ in the first round. $c_1$ then rejects $s_3$, and $s_3$ will then apply to $c_2$. The final matching is $\{c_1:s_2, c_2: s_3, c_3: s_1\}$, with pairs $(s_1,c_1)$, $(s_1,c_2)$, $(s_2, c_3)$ and $(s_3, c_3)$ having positive probability to be blocks. Then, the corresponding $\mathsf{ProS}$ equals to $2/11$.

    On the contrary, if applying the LOICV method, when $c_1$ rejects $s_3$ in the first round, $s_3$ will then apply to $c_3$ instead because $\tilde{q}_{s_3}(c_2, \{c_1\}) = (2/5)$ and $\tilde{q}_{s_3}(c_3, \{c_1\}) = (3/5)$. $c_3$ then has choice to reject $s_1$, $c_1$ then accepts $s_1$, and $c_3$ can get its best choice $s_2$. The final matching is $\{c_1:s_1, c_2: s_3, c_3: s_2\}$, with no pairs having positive probability to be blocks. Then, the $\mathsf{ProS}$ equals to $1$. We can observe that, compared to the LOCV method, the outcome for $s_3$, who actually has a different proposing order, does not change. Lastly, if applying the HEUF method, the proposing orders of the students are the same as those in the LOICV method, since $\mathbb{E}[u_{s_3}(c_1)]=0.55 $ $ > \mathbb{E}[u_{s_3}(c_3)]=0.35 > \mathbb{E}[u_{s_3}(c_2)]=0.3$. Hence, the matching result of the HEUF method is the same as that of the LOICV method. Similarly, if applying the HERF method, the proposing orders of the students are also the same, which leads to the same result.

    We can observe that in the first instance, applying the LOICV, HEUF, or HERF method is the best choice for a higher $\mathsf{ProS}$. Next, we will give the second instance where the LOCV method becomes better than LOICV, which suggests iteratively updating the comparison vector is not always a better choice. The preferences of colleges in the second instance are $s_3\succ_{c_1} s_1 \succ_{c_1} s_2$ and $s_2\succ_{c_2/c_3} s_3 \succ_{c_2/c_3} s_1$,
    and the utility functions of students are
    \begin{itemize}
        \item $s_1$: $u_{s_1}^{f_1}(c_1) = 0.9,\ u_{s_1}^{f_1}(c_2) = 0.7,\ u_{s_1}^{f_1}(c_3) = 0.75$;

        $\quad\ u_{s_1}^{f_2}(c_1) = 0.1,\ u_{s_1}^{f_2}(c_2) = 0.7,\ u_{s_1}^{f_2}(c_3) = 0.8$;
        \item $s_2$: $u_{s_2}^{f_{1/2}}(c_1) = 0.9,\ u_{s_2}^{f_{1/2}}(c_2) = 0.5,\ u_{s_2}^{f_{1/2}}(c_3) = 0.1$;
        \item $s_3$: $u_{s_3}^{f_{1/2}}(c_1) = 0.5,\ u_{s_3}^{f_{1/2}}(c_2) = 0.1,\ u_{s_3}^{f_{1/2}}(c_3) = 0.9$.
    \end{itemize}
    Notice that for student $s_2$ and $s_3$, their preferences are actually certain because the utility functions for all features are the same. For student $s_1$, if applying the LOCV method, she will first propose to $c_3$, and then propose to $c_1$ if being rejected by $c_3$. The final matching is $\{c_1:s_1, c_2: s_2, c_3: s_3\}$, with no pairs having positive probability to be blocks ($\mathsf{ProS} = 1$). On the contrary, if applying the LOICV method, $s_1$ will propose to $c_2$ if being rejected by the first choice $c_3$. Then, the final matching is $\{c_1:s_2, c_2: s_1, c_3: s_3\}$, with the pair $(s_1,c_1)$ having probability of $1/4$ to be a block ($\mathsf{ProS} = 3/4$). Lastly, the HEUF and HERF share the same proposing order of students in the LOICV method, so they have the same result.

    In the above two instances, the LOICV, HEUF and HERF methods share the same results. For the HEUF, it always shares the same result with the LOICV as long as the probability distributions satisfy certain conditions (see Theorem~\ref{thm:heufICR}). However, it is not true for the HERF method, as we can see in the third instance. For this last instance, the preferences of colleges are
    \begin{itemize}
        \item $c_1/c_3$: $s_1\succ_{c_1/c_3} s_3 \succ_{c_1/c_3} s_2$; $\quad c_2$: $s_3\succ_{c_2} s_2 \succ_{c_2} s_1$,
    \end{itemize}
    and the utility functions of students are
    \begin{itemize}
        \item $s_{1/2}$: $u_{s_{1/2}}^{f_1}(c_1) = 0.25,\ u_{s_{1/2}}^{f_1}(c_2) = 0.3,\ u_{s_{1/2}}^{f_1}(c_3) = 0.8$;

        $\quad\quad u_{s_{1/2}}^{f_2}(c_1) = 0.4,\ u_{s_{1/2}}^{f_2}(c_2) = 0.3,\ u_{s_{1/2}}^{f_2}(c_3) = 0.7$;
        \item $s_3$: $u_{s_3}^{f_1}(c_1) = 0.1,\ u_{s_3}^{f_1}(c_2) = 1.0,\ u_{s_3}^{f_1}(c_3) = 0.8$;

        $\quad\ u_{s_3}^{f_2}(c_1) = 1.0,\ u_{s_3}^{f_2}(c_2) = 0.2,\ u_{s_3}^{f_2}(c_3) = 0.5$.
    \end{itemize}
    For students $s_1$ and $s_2$, proposing order is always $c_3$, $c_1$ and $c_2$ in all methods. For student $s_3$, she will first propose to $c_3$ with the LOCV or the LOICV method because both $\mathrm{Pr}[c_3\succeq_{s_3}^{w_{s_3}}c_1] = 7/12$ and $\mathrm{Pr}[c_3\succeq_{s_3}^{w_{s_3}}c_2] = 3/5$ is larger than $1/2$. If being rejected by $c_3$, $s_3$ will then propose to $c_1$ in the LOCV, and propose to $c_2$ in the LOICV. Hence, for LOCV, the final matching is $\{c_1:s_3, c_2: s_2, c_3: s_1\}$, with the pair $(s_3,c_2)$ having probability of $9/17$ to be a block ($\mathsf{ProS} = 8/17$), while for LOICV, the final matching is $\{c_1:s_2, c_2: s_3, c_3: s_1\}$, with the pair $(s_3,c_1)$ having probability of $8/17$ to be a block ($\mathsf{ProS} = 9/17$). The result of HEUF is still the same as the LOICV.

    If applying the HERF method, $s_3$ will not first propose $c_3$ because $\mathrm{Pr}[c_3\succeq_{s_3}^{w_{s_3}}c_1 \text{ and } c_3\succeq_{s_3}^{w_{s_3}}c_2] = \mathrm{Pr}[5/12\leq w_{s_3}\leq 3/5] = 11/60$, which suggests that $c_3$ has the smallest probability to be the favorite college of $s_3$. In this case, $s_3$ will first propose to $c_1$, and finally lead to the same matching result as the LOCV method.
\end{example}

From the three instances in this example, we can observe that none of the four methods dominates the others, i.e., for each method, there exist instances where it is not the best choice. Although in special cases we cannot say that one method is always better than another, one may still be curious about their average performance across a large number of instances. Therefore, we conduct a simple simulation experiment to get insight on how four methods perform differently in average rather than the worst-case comparison.

In the experiment, we set $|F| = 2$ and assume uniform distributions for weight vectors (as stated in Theorem~\ref{thm:equivalent}, it actually can represent the results for a class of distributions). For utility functions, each value of $u_s^f(c)$ is uniformly sampled from $(0,1)$. We sample as large amount of instances as possible to cover all kinds of cases. The number of students and colleges are relatively not so large since we have to enumerate all possible matching results to find an optimal solution, which is computationally hard.

The results of the experiment are summarized as box-plots in Figure~\ref{fig:experiment}. From the experiment, we can observe that the method of LOICV and HERF outperforms the method of LOCV when the seats of colleges are very rich for students. When the number of students increases, the outliers that represent worse cases for all methods become rare. It suggest that all the methods can perform well in a random sampled instance, and have high chance to achieve an optimal matching that maximizes the probability of stability.

\begin{figure*}[h]
    \centering
    \includegraphics[width=\linewidth]{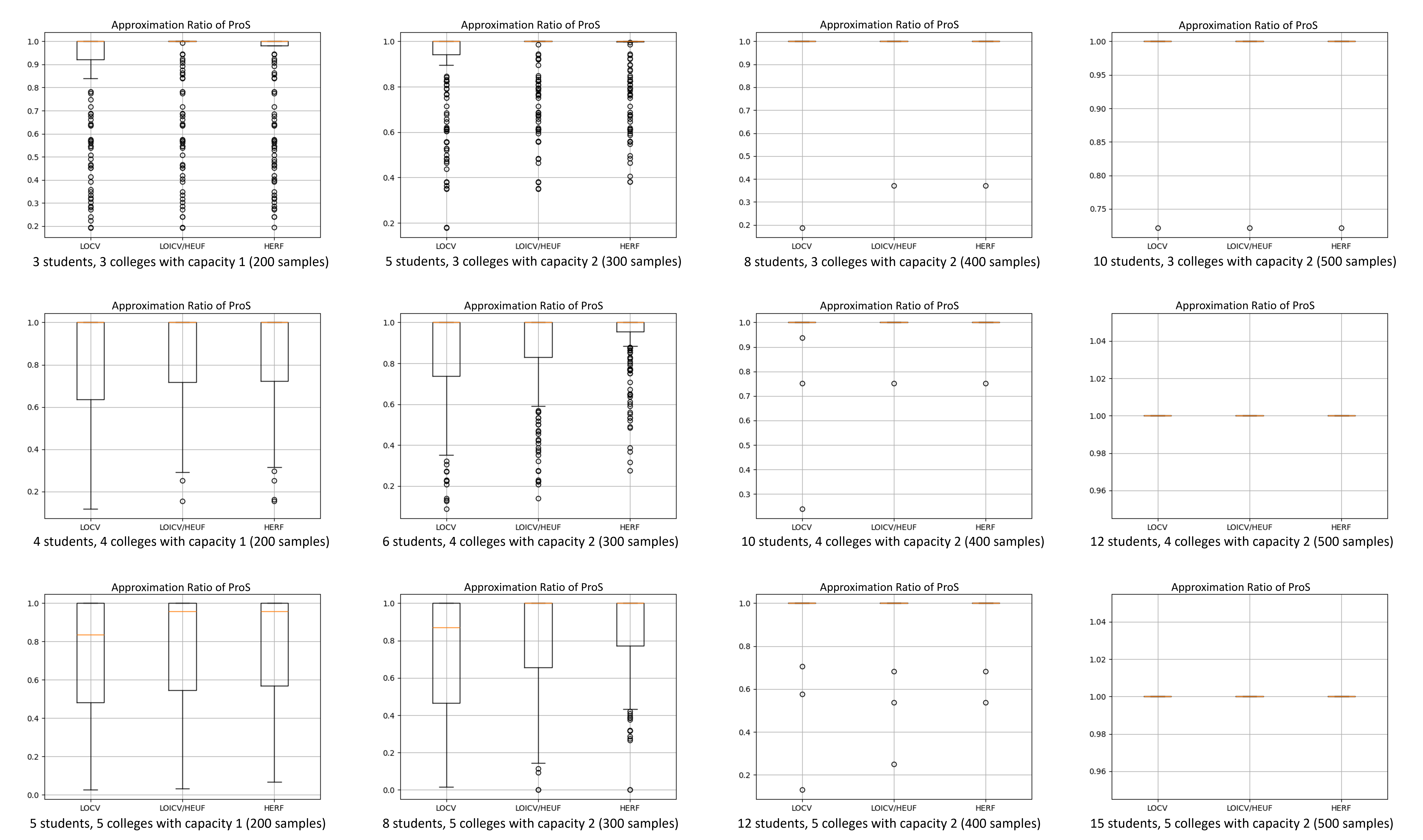}
    \caption{Box-plots of the approximation ratio of $\mathsf{ProS}$ in all sampled instances.}
    \label{fig:experiment}
\end{figure*}
\section{Properties of Methods}\label{sec:twofeatures}
In this section, we analyze the properties of these four methods.

\subsection{Probability of Stability}
For the probability of stability, only the method of HERF can have a positive approximation ratio.

\begin{theorem}\label{thm:pros0s}
    For GDA with the HEUF, LOCV, and LOICV methods, the worst-case approximation ratios to the optimal probability of being stable are all 0 even when the number of features $|F| = 2$.
\end{theorem}
\begin{proof}
    We prove the statement by showing a worst-case instance with three students and three colleges. The utility functions of three students are listed as follows.
    \begin{itemize}
        \item $s_1$: \begin{align*}
            & u_{s_1}^{f_1}(c_1) = 0.75, &&u_{s_1}^{f_1}(c_2) = 0.5, &&u_{s_1}^{f_1}(c_3) = 0.55;\\
            & u_{s_1}^{f_2}(c_1) = 0.55, &&u_{s_1}^{f_2}(c_2) = 0.25, &&u_{s_1}^{f_2}(c_3) = 0.1;
        \end{align*}
        \item $s_1$: \begin{align*}
            & u_{s_2}^{f_1}(c_1) = 0.1 + 1.5\delta + 3\epsilon, &&u_{s_2}^{f_1}(c_2) = 0.1 + \delta + 2\epsilon, \\ 
            & u_{s_2}^{f_1}(c_3) = 0.1; && u_{s_2}^{f_2}(c_1) = 0.1, \\
            &u_{s_2}^{f_2}(c_2) = 0.1 + 0.5\delta + 2\epsilon, &&u_{s_2}^{f_2}(c_3) = 0.1 + 1.5\delta + 2\epsilon;
        \end{align*}
        \item $s_3$: $u_{s_1}^{f_{1/2}}(c_1) = 0.3,\ u_{s_1}^{f_{1/2}}(c_2) = 0.2,\ u_{s_1}^{f_{1/2}}(c_3) = 0.1$.
    \end{itemize}
    Here, $\delta$ and $\epsilon$ are any positive values, and $\epsilon\rightarrow 0$. We assume that the probability distributions of weights for all three students are uniform distributions over all possible weights. On the other hand, the preferences of three colleges are listed as follows.
    \begin{itemize}
        \item $c_1$: $s_2\succ_{c_1} s_3\succ_{c_1} s_1$;
        \item $c_2$: $s_2\succ_{c_2} s_3\succ_{c_2} s_1$;
        \item $c_3$: $s_3\succ_{c_3} s_2\succ_{c_3} s_1$.
    \end{itemize}
    Here, we suppose the capacities of all colleges are one. In this instance, no matter choosing the HEUF, LOCV or LOICV method, the final matching $\pi$ is always $\{c_1:s_3, c_2: s_2, c_3: s_1\}$, with the probability of being stable as
    \[ \mathsf{ProS}(\pi) = \mathrm{Pr}[c_2\succeq_{s_2}^{w_{s_2}}c_1 \text{ and } c_2\succeq_{s_2}^{w_{s_2}}c_3] = \frac{\epsilon(\delta+2\epsilon)}{(\delta+\epsilon)(\delta+3\epsilon)}. \]
    However, an optimal matching $\pi'$ should be $\{c_1:s_2, c_2: s_3, c_3: s_1\}$, with the probability of being stable as
    \[ \mathsf{ProS}(\pi') = \mathrm{Pr}[c_1\succeq_{s_2}^{w_{s_2}}c_2 \text{ and } c_1\succeq_{s_2}^{w_{s_2}}c_3] = \frac{\delta+2\epsilon}{2\delta+6\epsilon}. \]
    That means the worst-case approximation ratio for all three methods is
    \[ \alpha < \frac{\mathsf{ProS}(\pi)}{\mathsf{ProS}(\pi')} = \frac{2\epsilon}{\delta+\epsilon}\rightarrow0 \quad\text{as}\quad \epsilon\rightarrow 0. \]
    This concludes the proof.
\end{proof}

\begin{theorem}\label{thm:HERFalpha}
    The GDA with the HERF method is $\left(\frac{1}{n}\right)^n$-optimal. 
\end{theorem}
\begin{proof}
    For each student $s\in N$, when she is finally matched to a college $c$ under the HERF method, we claim that the probability of $c$ being ranked the first by $s$ among all colleges that have not rejected $s$ is at least $1/m_s^u$, where $m_s^u$ denotes the number of colleges which have not rejected $s$. This claim can be shown by contradiction. If the probability is less than $1/m_s^u$, then total probability of each college being ranked the first among all colleges that have not rejected $s$ will be less than $(1/m_s^u)\cdot m_s^u < 1$ because $c$ has the highest probability of being ranked the first. However, the total probability should be at least one\footnote{It might be greater than one if there exist aggregated preferences that have a tie at the head.}. Hence, the claim must be satisfied.

    Then, for any aggregated preference where $c$ is ranked the first by $s$ among all colleges that have not rejected $s$, $c$ cannot make a block with any college: those have rejected $s$ must have full seats where each of them is better than $s$; others have not rejected $s$ are less preferred by $s$. Therefore, the approximation ratio is at least
    \[ \alpha \geq \mathsf{ProS}(\pi) \geq \prod_{s\in N}\frac{1}{m_s^u} \geq \prod_{s\in N} \frac{1}{n} = \left(\frac{1}{n}\right)^n. \]
\end{proof}

\begin{proposition}
    The worst-case approximation ratio characterized in Theorem~\ref{thm:HERFalpha} is tight even when the number of features $|F| = 2$.
\end{proposition}

\begin{proof}
    We illustrate the tightness of the approximation ratio by showing an instance with $|F|=2$ that meets the worst case scenario. Consider an instance with $n$ students and $m=n$ colleges, and each college has a capacity of one. For each college $c_j\in M$, suppose $s_j$ is the least preferred student, and $s_{j+1}$ ($s_1$ if $j=n$) is the most preferred student according to its preference. On the other hand, for each student $c_i\in N$, the probability distribution $\mu_s$ is a uniform distribution over all possible weight values, and her utility function is as follows.
    \begin{itemize}
        \item $u_{s_i}^{f_1}(c_j) = \nu_j$ for all $j\neq i$, and $u_{s_i}^{f_1}(c_i) = \nu_i + \epsilon$;
        \item $u_{s_i}^{f_1}(c_j) = \nu_{n-j+1}$ for all $j\neq i$, and $u_{s_i}^{f_1}(c_i) = \nu_{n-i+1} + \epsilon$,
    \end{itemize}
    where the values of $\nu$ are defined as: (i) $\nu_1 = 0$, and (ii) $\nu_k = \nu_{k-1} + (n-k+1)\delta$ for all $1<k\leq n$. Here, $0<\delta\leq \frac{2}{n(n-1)}$ and $\epsilon$ is a very small positive number. Under this utility function, it partitions the intervals of $w_{s_i}^{f_1}$ as follows.
    \begin{itemize}
        \item For $1\leq j<i-1$, $c_j$ will be ranked the first in the aggregated preference of $s_i$ when $\frac{j-1}{n}\leq w_{s_i}^{f_1}\leq \frac{j}{n}$;
        \item For $j = i - 1$, $c_j$ will be ranked the first in the aggregated preference of $s_i$ when $\frac{i-2}{n}\leq w_{s_i}^{f_1}\leq \frac{i-1}{n}-\frac{\epsilon}{n\delta}$;
        \item For $j = i$, $c_j$ will be ranked the first in the aggregated preference of $s_i$ when $\frac{i-1}{n}-\frac{\epsilon}{n\delta} \leq w_{s_i}^{f_1} \leq \frac{i}{n} + \frac{\epsilon}{n\delta}$;
        \item For $j = i + 1$, $c_j$ will be ranked the first in the aggregated preference of $s_i$ when $\frac{i}{n} + \frac{\epsilon}{n\delta} \leq w_{s_i}^{f_1} \leq \frac{i+1}{n}$;
        \item For $i+1< j\leq n$, $c_j$ will be ranked the first in the aggregated preference of $s_i$ when $\frac{j-1}{n}\leq w_{s_i}^{f_1}\leq \frac{j}{n}$.
    \end{itemize}
    Since $w_{s_i}^{f_1}$ is a uniformly distributed random variable over $[0,1]$, $c_i$ then has the highest probability to be ranked first by $s_i$. Then, the final result $\pi$ of the HERF method matches $s_i$ to $c_i$ for every $i$. However, since $c_i$ likes $s_i$ the least, then each pair $(s_i, c_j)$ with $j\neq i$ is a potential block. The final probability of stability is
    \[ \mathsf{ProS}(\pi) = \prod_{i=1}^n \left( \frac{i}{n} + \frac{\epsilon}{n\delta} - \left(\frac{i-1}{n}-\frac{\epsilon}{n\delta}\right) \right) = \left( \frac{1}{n} + \frac{2\epsilon}{n\delta} \right)^n, \]
    which approaches to $(1/n)^n$ when $\epsilon\rightarrow 0$. On the other hand, consider an alternative matching $\pi'$ such that $s_i$ is matched to $c_{i-1}$ ($c_n$ if $i=1$). Since each college gets its favorite student, then $\mathsf{ProS}(\pi') = 1$. Therefore, the worst-case approximation ratio is no larger than $(1/n)^n$.
\end{proof}

\subsection{Incentive Compatibility}
First, for all four methods, they all satisfy the minimal requirement of incentive compatibility.

\begin{theorem}\label{thm:ICC}
    The GDA with the LOCV, LOICV, HERF, or HEUF method is IC-C.
\end{theorem}
\begin{proof}
    For any student $s\in N$, whichever one of the four methods is selected, we claim that any misreport cannot make her be accepted by a college that has rejected her when she truthfully reports. This can be shown by contradiction. Suppose that the proposing order when $s$ truthfully reports her private information is $o_s$, and the alternative proposing order when she misreports her information is $o_s'$. If with the proposing order $o_s'$, $s$ can be accepted by a college that has rejected her with the proposing order $o_s$, then we can construct an instance of traditional school choice with $o_s$ as preferences of students, where $s$ can misreport her preference as $o_s'$ to be matched with a better college under the student-proposing deferred acceptance algorithm. It contradicts to the fact the student-proposing DA algorithm is truthful for students. Hence, we only need to consider colleges that has not been proposed by $s$ after the algorithm terminates. In other words, it is sufficient to show that, in each round, a student $s$ proposes to a college $c$ indicates that there does not exist another college $c'$ that has not rejected $s$, such that $\mathrm{Pr}[c'\succ_s^{w_s} c] = 1$.

    For LOCV, it can be shown by contradiction. Suppose in one round, a student $s$ proposes to a college $c$, but there is another college $c'$ that has not rejected $s$ such that $\mathrm{Pr}[c'\succ_s^{w_s} c] = 1$. Then, for any other college $c''$, we have $\mathrm{Pr}[c'\succeq_s^{w_s} c''] \geq \mathrm{Pr}[c\succ_s^{w_s} c' \text{ and } c\succeq_s^{w_s} c''] = \mathrm{Pr}[c\succeq_s^{w_s} c'']$. This implies that $q_s(c)$ must be lexicographically smaller than $q_s(c')$, which contradicts the fact that $s$ proposes to $c$. Hence, such a college $c'$ does not exist.

    For LOICV, it can also be shown by contradiction. Suppose in one round, a student $s$ proposes to a college $c$, but there is another college $c'$ that has not rejected $s$ such that $\mathrm{Pr}[c'\succ_s^{w_s} c] = 1$. Then, the first element of comparison vector in this round $\tilde{q}_s(c, R_s)$ must be zero. By the definition of LOICV, it implies that for each college $c'\notin R_s$, the first element of comparison vector in this round is zero. This means among all colleges that have not rejected $s$, each of them is less preferred by another college with probability 1. Then, an aggregated preference of $s$ must contain a cycle, which is impossible since the aggregated preferences is defined by expected utilities. Hence, such a college $c'$ does not exist.

    For HERF, the existence of $c'$ implies that the probability of $c$ is the most preferred college among $M\setminus R_s$ will be 0. However, since an aggregated preference of $s$ does not contain any cycles, there must be at least one college has positive probability to be most preferred among any subset of colleges. Hence, such a college $c'$ does not exist.

    For HEUF, $\mathrm{Pr}[c'\succ_s^{w_s} c] = 1$ means that for any possible $w_s$, the weighted utility of $c'$ is larger than that of $c$: $\sum_{f\in F}w_s^fu_s^f(c') > \sum_{f\in F} w_s^fu_s^f(c)$. This directly implies that the expected utility of $c$ also must be larger that of $c$, which contradicts the definition of the HEUF method. Hence, such a college $c'$ does not exist.

    In summary, all of the four methods are IC-C.
\end{proof}

Additionally, for the method of LOICV, it can further satisfy IC-R when $|F| = 2$.

\begin{lemma}\label{lem:transitivity}
    When $|F| = 2$, then for any student $s\in N$, and any three different colleges $c_i$, $c_j$, $c_k\in M$, if $\mathrm{Pr}[c_i \succeq_s^{w_s} c_j]\geq \frac{1}{2}$ and $\mathrm{Pr}[c_j \succeq_s^{w_s} c_k]\geq \frac{1}{2}$, we must have $\mathrm{Pr}[c_i \succeq_s^{w_s} c_k]\geq \frac{1}{2}$.
\end{lemma}
\begin{proof}
    Since $\mathrm{Pr}[c_i \succ_s^{w_s} c_j]\geq \frac{1}{2}$, then either $u_s^{f_1}(c_i) \geq u_s^{f_1}(c_j)$ or $u_s^{f_2}(c_i) \geq u_s^{f_2}(c_j)$ must be satisfied. W.l.o.g., we assume that $u_s^{f_1}(c_i) \geq u_s^{f_1}(c_j)$. Similarly, if $u_s^{f_1}(c_j) \leq u_s^{f_1}(c_k)$, then we must have $u_s^{f_2}(c_j) \geq u_s^{f_2}(c_k)$. Then, based on the relationships among utilities of three colleges, we have totally the following cases.

    \noindent\textbf{Case I}: $u_s^{f_1}(c_i) \geq u_s^{f_1}(c_j) \geq u_s^{f_1}(c_k)$. In this case, by Lemma~\ref{lem:probformula}, we have $\mathrm{Pr}[c_i\succeq_s^{w_s} c_j] = 1$ or $\mathrm{Pr}[c_i\succeq_s^{w_s} c_j] = \mathrm{Pr}[w_s^{f_1}\geq \eta_s(c_i,c_j)]$, and similarly, we also have $\mathrm{Pr}[c_j\succeq_s^{w_s} c_k] = 1$ or $\mathrm{Pr}[c_j\succeq_s^{w_s} c_k] = \mathrm{Pr}[w_s^{f_1}\geq \eta_s(c_j,c_k)]$. Hence, the probability of $\mathrm{Pr}[c_i\succeq_s^{w_s}c_j \text{ and } c_j\succeq_s^{w_s}c_k]$ has four possible values: 1) 1 if both probabilities are 1, 2) $\mathrm{Pr}[c_j\succeq_s^{w_s} c_k]\geq 1/2$, if $\mathrm{Pr}[c_i\succeq_s^{w_s} c_j]=1$, 3) $\mathrm{Pr}[c_i\succeq_s^{w_s} c_j]\geq 1/2$, if $\mathrm{Pr}[c_j\succeq_s^{w_s} c_k] = 1$, or 4) $\mathrm{Pr}[w_s^{f_1}\geq \eta_s(c_i,c_j) \text{ and }w_s^{f_1}\geq \eta_s(c_j,c_k)] = \mathrm{Pr}[w_s^{f_1}\geq \min\{\eta_s(c_i,c_j), \eta_s(c_j,c_k)\}]$, which finally equals to one of the two probabilities and greater than $1/2$. Therefore, $\mathrm{Pr}[c_i \succeq_s^{w_s} c_k]\geq \mathrm{Pr}[c_i\succeq_s^{w_s}c_j \text{ and } c_j\succeq_s^{w_s}c_k] \geq \frac{1}{2}$.
    
    \noindent\textbf{Case II}: $u_s^{f_1}(c_i) \geq u_s^{f_1}(c_k) \geq u_s^{f_1}(c_j)$ and $u_s^{f_2}(c_i) \geq u_s^{f_2}(c_j) \geq u_s^{f_2}(c_k)$. In this case, since both $u_s^{f_1}(c_i) \geq u_s^{f_1}(c_k)$ and $u_s^{f_2}(c_i) \geq u_s^{f_2}(c_k)$ are satisfied, we have $\mathrm{Pr}[c_i \succeq_s^{w_s} c_k] = 1$.
    
    \noindent\textbf{Case III}: $u_s^{f_1}(c_i) \geq u_s^{f_1}(c_k) \geq u_s^{f_1}(c_j)$ and $u_s^{f_2}(c_j) \geq u_s^{f_2}(c_i) \geq u_s^{f_2}(c_k)$. Same to Case II, we have $\mathrm{Pr}[c_i \succeq_s^{w_s} c_k] = 1$ in this case.
        
    \noindent\textbf{Case IV}: $u_s^{f_1}(c_i) \geq u_s^{f_1}(c_k) \geq u_s^{f_1}(c_j)$ and $u_s^{f_2}(c_j) \geq u_s^{f_2}(c_k) \geq u_s^{f_2}(c_i)$. In this case, we can observe that $\Delta_s^f(c_i,c_j) = \Delta_s^f(c_i,c_k) + \Delta_s^f(c_j,c_k)$ for $f\in \{f_1, f_2\}$. Hence, by Lemma~\ref{lem:probformula}, we have $\eta_s(c_i,c_j)$ $= \frac{\Delta_s^{f_1}(c_i,c_k)+\Delta_s^{f_2}(c_i,c_k)}{\Delta_s^{f_1}(c_i,c_j)+\Delta_s^{f_2}(c_i,c_j)}\eta_s(c_i,c_k) + \frac{\Delta_s^{f_1}(c_j,c_k) + \Delta_s^{f_2}(c_j,c_k)}{\Delta_s^{f_1}(c_i,c_j)+\Delta_s^{f_2}(c_i,c_j)}\eta_s(c_j,c_k)$ (for the special case where $\Delta_s^{f_1}(c_i,c_j)+\Delta_s^{f_2}(c_i,c_j) = 0$, the values of three colleges are the same. Then, the statement is automatically true). This indicates that $\eta_s(c_i,c_j)$ is a convex combination of $\eta_s(c_i,c_k)$ and $\eta_s(c_j,c_k)$, so one of $\eta_s(c_i, c_k) \leq \eta_s(c_i, c_j)\leq \eta_s(c_j,c_k)$ and $\eta_s(c_i, c_k) \geq \eta_s(c_i, c_j)\geq \eta_s(c_j,c_k)$ must be satisfied. For the former case, we have $\mathrm{Pr}[c_i \succeq_s^{w_s} c_k] = \mathrm{Pr}[w_s^{f_1}\geq \eta_s(c_i,c_k)] \geq \mathrm{Pr}[w_s^{f_1}\geq \eta_s(c_i,c_j)] = \mathrm{Pr}[c_i \succeq_s^{w_s} c_j] \geq \frac{1}{2}$. For the latter case, since we have $\mathrm{Pr}[c_i \succeq_s^{w_s} c_j] = \mathrm{Pr}[w_s^{f_1}\geq \eta_s(c_i,c_j)] \geq 1/2$ and $\mathrm{Pr}[c_j \succeq_s^{w_s} c_k] = \mathrm{Pr}[w_s^{f_1}\leq \eta_s(c_j,c_k)] \leq 1/2$, then it can only be $\eta_s(c_i,c_j) = \eta_s(c_j,c_k)$. According to the coefficients of the convex combination, this leads to the same utilities of college $c_i$ and $c_k$, and then $\mathrm{Pr}[c_i \succeq_s^{w_s} c_k] = 1$.
    
    \noindent\textbf{Case V}: $u_s^{f_1}(c_k) \geq u_s^{f_1}(c_i) \geq u_s^{f_1}(c_j)$ and $u_s^{f_2}(c_i) \geq u_s^{f_2}(c_j) \geq u_s^{f_2}(c_k)$. By symmetry with Case I, this case is definitely true.
    
    \noindent\textbf{Case VI}: $u_s^{f_1}(c_k) \geq u_s^{f_1}(c_i) \geq u_s^{f_1}(c_j)$ and $u_s^{f_2}(c_j) \geq u_s^{f_2}(c_i) \geq u_s^{f_2}(c_k)$. Similar to Case IV, here, $\eta_s(c_j,c_k)$ is a convex combination of $\eta_s(c_i,c_j)$ and $\eta_s(c_i,c_k)$, so one of $\eta_s(c_i, c_k) \leq \eta_s(c_j, c_k)\leq \eta_s(c_i,c_j)$ and $\eta_s(c_i, c_k) \geq \eta_s(c_j, c_k)\geq \eta_s(c_i,c_j)$ must be satisfied. For the former case, since we have $\mathrm{Pr}[c_i \succeq_s^{w_s} c_j] = \mathrm{Pr}[w_s^{f_1}\geq \eta_s(c_i,c_j)] \geq 1/2$ and $\mathrm{Pr}[c_j \succeq_s^{w_s} c_k] = \mathrm{Pr}[w_s^{f_1}\leq \eta_s(c_j,c_k)] \leq 1/2$, then it can only be $\eta_s(c_j,c_k) = \eta_s(c_i,c_j)$. According to the coefficients of the convex combination, this leads to the same utilities of college $c_i$ and $c_k$, and then $\mathrm{Pr}[c_i \succeq_s^{w_s} c_k] = 1$. For the latter case, we have $\mathrm{Pr}[c_i \succeq_s^{w_s} c_k] = \mathrm{Pr}[w_s^{f_1}\leq \eta_s(c_i,c_k)] \geq \mathrm{Pr}[w_s^{f_1}\leq \eta_s(c_j,c_k)] = \mathrm{Pr}[c_j \succeq_s^{w_s} c_k] \geq \frac{1}{2}$.
    
    \noindent\textbf{Case VII}: $u_s^{f_1}(c_k) \geq u_s^{f_1}(c_i) \geq u_s^{f_1}(c_j)$ and $u_s^{f_2}(c_j) \geq u_s^{f_2}(c_k) \geq u_s^{f_2}(c_i)$. In this case, since $\Delta_s^{f_1}(c_i,c_j) \leq \Delta_s^{f_1}(c_j,c_k)$ and $\Delta_s^{f_2}(c_i,c_j) \geq \Delta_s^{f_2}(c_j,c_k)$, then by Lemma~\ref{lem:probformula}, we have $\eta_s(c_i,c_j)\geq \eta_s(c_j,c_k)$. On the other hand, we have $\mathrm{Pr}[c_i \succeq_s^{w_s} c_j] = \mathrm{Pr}[w_s^{f_1}\geq \eta_s(c_i,c_j)] \geq 1/2$ and $\mathrm{Pr}[c_j \succeq_s^{w_s} c_k] = \mathrm{Pr}[w_s^{f_1}\leq \eta_s(c_j,c_k)] \leq 1/2$, then it can only be $\eta_s(c_j,c_k) = \eta_s(c_i,c_j)$. Hence, $\Delta_s^{f_1}(c_i,c_j) = \Delta_s^{f_1}(c_j,c_k)$ and $\Delta_s^{f_2}(c_i,c_j) = \Delta_s^{f_2}(c_j,c_k)$, which means $c_i$ and $c_k$ have the same utilities, and then $\mathrm{Pr}[c_i \succeq_s^{w_s} c_k] = 1$.

    Taking all the above together, we can conclude that the statement of lemma is true.
\end{proof}

\begin{theorem}\label{thm:LOICV-ICR}
    The GDA with the LOICV method is IC-R when the number of features $|F| = 2$.
\end{theorem}

\begin{proof}
    First, as stated in the proof of Theorem~\ref{thm:ICC}, for any $s\in N$, no misreport can make her be accepted by a college that rejected her when she reported truthfully, so we only need to consider colleges that have not been proposed to by $s$ after the algorithm terminates.

    By Lemma~\ref{lem:transitivity}, we can make a corollary that for any student $s$, and any subset of colleges $M'\subseteq M$ with $|M'| > 1$, there must exist a college $c\in M'$, such that $\mathrm{Pr}[c \succeq_s^{w_s} c'] \geq \frac{1}{2}$, for any other $c'\in M'$. The statement is obviously true for $|M'| = 2$. For $|M'|>2$, it can be shown by contradiction. Suppose that for any $c\in M'$, there is at least one another college $c'\in M'$ such that $\mathrm{Pr}[c' \succ_s^{w_s} c] > \frac{1}{2}$. Then, we can construct a directed graph where nodes are colleges in $M'$, and an edge $(c,c')$ appears if and only if $\mathrm{Pr}[c' \succ_s^{w_s} c] > \frac{1}{2}$. Since each node has at least one out-edge appoints to the other node, there must be at least one cycle in the graph. Denote the cycle as $(c_1,c_2,\dots,c_k,c_1)$, and we have $\mathrm{Pr}[c_{i} \succ_s^{w_s} c_{i+1}] > \frac{1}{2}$ for $1\leq i< k$. By Lemma~\ref{lem:transitivity}, we must also have $\mathrm{Pr}[c_1 \succeq_s^{w_s} c_k] \geq \frac{1}{2}$, which contradicts to $\mathrm{Pr}[c_{k} \succ_s^{w_s} c_{1}] > \frac{1}{2}$.

    Then, according to the process of the LOICV method, in each round, the college proposed by a student $s$ must be one of the colleges described above, with $M'$ as the set of colleges that have not rejected $s$, since the comparison vectors are in the ascending order of elements. Therefore, for any college $c'$ that $s$ can be matched with by misreporting, compared to the college $c$ she can be matched without misreporting, there must be $\mathrm{Pr}[c' \succ_s^{w_s} c] = 1 - \mathrm{Pr}[c \succeq_s^{w_s} c'] \leq \frac{1}{2}$. This shows that the method is IC-R.
\end{proof}

However, when $|F|\geq 3$, it can no longer be IC-R because the transitivity characterized in Lemma~\ref{lem:transitivity} might be violated (see Example~\ref{ex:notrans} for details). Finally, for the method of HEUF, it can also be IC-R for a certain class of probability distributions. 

\begin{theorem}\label{thm:heufICR}
    The GDA with the HEUF method is IC-R if and only if for any student $s\in N$, $\mathrm{Pr}[w_s^{f_1}\geq \mathbb{E}_{w_s\sim \mu_s}[w_s^{f_1}]] = \mathrm{Pr}[w_s^{f_1}\leq \mathbb{E}_{w_s\sim \mu_s}[w_s^{f_1}]]$, when the number of features $|F| = 2$.
\end{theorem}
\begin{proof}
    ``$\Rightarrow$'': If the HEUF method is IC-R in a specific domain of instances, then for any student $s\in N$, and any two colleges $c_i$, $c_j\in M$, such that $\mathrm{Pr}[c_i\succ_s^{w_s} c_j] > 1/2$, $s$ must propose $c_i$ earlier than proposing $c_j$; otherwise, $s$ will have an incentive to misreport if she is finally accepted by $c_j$. By Lemma~\ref{lem:probformula}, there are three cases when $\mathrm{Pr}[c_i\succ_s^{w_s} c_j] > 1/2$.
    
    (I). If $u_s^{f_1}(c_i) > u_s^{f_1}(c_j)$ and $u_s^{f_2}(c_i) > u_s^{f_2}(c_j)$, then the expected utility of $c_i$ must be larger than that of $c_j$. In this case, $s$ will always propose $c_i$ earlier than $c_j$ in the HEUF method.

    (II). If $u_s^{f_1}(c_i) > u_s^{f_1}(c_j)$ and $u_s^{f_2}(c_i) \leq u_s^{f_2}(c_j)$, then $\mathrm{Pr}[c_i\succ_s^{w_s} c_j] = \mathrm{Pr}[w_s^{f_1} > \eta_s(c_i,c_j)] > 1/2$. On the other hand, if $s$ first proposes $c_i$ in the HEUF method, it must hold that
    \[ \mathbb{E}[w_s^{f_1}u_s^{f_1}(c_i) + (1-w_s^{f_1})u_s^{f_2}(c_i)] \geq \mathbb{E}[w_s^{f_1}u_s^{f_1}(c_j) + (1-w_s^{f_1})u_s^{f_2}(c_j)], \]
    and we can rearrange it as $\mathbb{E}[w_s^{f_1}]\geq \eta_s(c_i,c_j)$. This means for any value $x\in[0,1]$, $\mathrm{Pr}[w_s^{f_1} > x] > 1/2$ implies that $\mathbb{E}[w_s^{f_1}]\geq x$, i.e, the value of $\mathbb{E}[w_s^{f_1}]$ must be located right to the ``median''. Hence, $\mathrm{Pr}[w_s^{f_1}\leq \mathbb{E}[w_s^{f_1}]] \geq \mathrm{Pr}[w_s^{f_1}\geq \mathbb{E}[w_s^{f_1}]]$.

    (III). If $u_s^{f_1}(c_i) \leq u_s^{f_1}(c_j)$ and $u_s^{f_2}(c_i) > u_s^{f_2}(c_j)$, then similar to the case (II), we can derive that for any value $x\in[0,1]$, $\mathrm{Pr}[w_s^{f_1} < x] > 1/2$ implies that $\mathbb{E}[w_s^{f_1}]\leq x$, i.e, the value of $\mathbb{E}[w_s^{f_1}]$ must be located left to the ``median''. Hence, $\mathrm{Pr}[w_s^{f_1}\leq \mathbb{E}[w_s^{f_1}]] \leq \mathrm{Pr}[w_s^{f_1}\geq \mathbb{E}[w_s^{f_1}]]$.

    Taking all together, we can conclude that $\mathrm{Pr}[w_s^{f_1}\geq \mathbb{E}_{w_s\sim \mu_s}[w_s^{f_1}]]$ $= \mathrm{Pr}[w_s^{f_1}\leq \mathbb{E}_{w_s\sim \mu_s}[w_s^{f_1}]]$.

    ``$\Leftarrow$'': We prove the other side by showing the HEUF method is equivalent to the LOICV in the given cases. According to the proof of Theorem~\ref{thm:LOICV-ICR}, in each round, a student $s$ will propose to a college $c$ such that $\mathrm{Pr}[c\succeq_s^{w_s}c'] \geq 1/2$ for any other $c'$ that has not rejected $s$. By Lemma~\ref{lem:probformula}, there are three cases when $\mathrm{Pr}[c\succeq_s^{w_s}c'] \geq 1/2$.

    (I). If $u_s^{f_1}(c_i) \geq u_s^{f_1}(c_j)$ and $u_s^{f_2}(c_i) \geq u_s^{f_2}(c_j)$, then we also have the expected utility of $c_i$ must be no less than that of $c_j$.

    (II). If $u_s^{f_1}(c_i) \geq u_s^{f_1}(c_j)$ and $u_s^{f_2}(c_i) \leq u_s^{f_2}(c_j)$, then $\mathrm{Pr}[c_i\succeq_s^{w_s} c_j] = \mathrm{Pr}[w_s^{f_1} \geq \eta_s(c_i,c_j)] \geq 1/2$. Since $\mathrm{Pr}[w_s^{f_1}\leq \mathbb{E}[w_s^{f_1}]] = \mathrm{Pr}[w_s^{f_1}\geq \mathbb{E}[w_s^{f_1}]]$, we have $\mathbb{E}[w_s^{f_1}] \geq \eta_s(c_i,c_j)$, which further implies that the expected utility of $c_i$ must be no less than that of $c_j$.

    (III). If $u_s^{f_1}(c_i) \leq u_s^{f_1}(c_j)$ and $u_s^{f_2}(c_i) \geq u_s^{f_2}(c_j)$, then similar to the case (II), we can derive that $\mathbb{E}[w_s^{f_1}] \leq \eta_s(c_i,c_j)$, which can also imply that the expected utility of $c_i$ must be no less than that of $c_j$ in this case.

    Taking all together, such a college $c$ must also have the largest expected utility among all colleges that have not rejected $s$. Therefore, we can conclude that the two methods are equivalent to each other in the given domain, and both IC-R.
\end{proof}
\section{Additional Issues}\label{sec:additional}
\subsection{Restrictions of ``mean=median''}
Theorem~\ref{thm:heufICR} highlights a special property of those probability distributions on the first feature whose expectation equals their median when $|F| = 2$. Such distributions naturally arise in many situations. For instance, a student may firmly believe that, for her, the first feature is about twice as important as the second one, yet it is indistinguishable whether being slightly more or slightly less than twice is preferable. Moreover, we will show that these distributions possess additional theoretical properties, as they can be reduced to simple uniform distributions from certain perspectives.
\begin{definition}
    We say two students $s$ and $s'$ are \emph{equivalent in welfare and pairwise preferences} if they satisfy all the following:
    \begin{itemize}
        \item for any college $c\in M$, $\mathbb{E}[u_s(c)] = \mathbb{E}[u_{s'}(c)]$;
        \item for any feature $f\in F$, and any two colleges $c_i$, $c_j\in M$, $u_s^f(c_i)\geq u_s^f(c_j)$ if and only if $u_{s'}^f(c_i)\geq u_{s'}^f(c_j)$;
        \item for any two colleges $c_i$, $c_j\in M$, $\mathrm{Pr}[c_i\succeq_s^{w_s}c_j] \geq \mathrm{Pr}[c_j\succeq_s^{w_s}c_i]$ if and only if $\mathrm{Pr}[c_i\succeq_{s'}^{w_{s'}}c_j] \geq \mathrm{Pr}[c_j\succeq_{s'}^{w_{s'}}c_i]$.
    \end{itemize}
\end{definition}

\begin{theorem}\label{thm:equivalent}
    When the number of features is $|F|=2$, then for a student $s$ with a continuous probability $\mu_s$ that satisfies $\mathrm{Pr}[w_s^{f_1}\leq \mathbb{E}[w_s^{f_1}]] = 1/2$, the following $s'$ is equivalent to $s$ in welfare and pairwise preferences:
    \begin{itemize}
        \item $\mu_{s'}$ is a uniform distribution over all possible weight vectors;
        \item let $A = \int_0^1 \mathrm{Pr}[w_s^{f_1}\leq w] \mathrm{d}w$, and then for any college $c\in M$, $u_{s'}^{f_1} = 2(1-A)u_s^{f_1}$ and $u_{s'}^{f_2} = 2Au_s^{f_2}$.
    \end{itemize}
\end{theorem}
\begin{proof}
    We first show that $s$ and $s'$ satisfy the conditions in the definition one by one.

    Firstly, for any college $c\in M$, we have
    \begin{align*}
        \mathbb{E}&[u_s(c)]  = \int_{0}^1 \mathrm{Pr}[w_s^{f_1}=w](w\cdot u_s^{f_1}(c) +(1-w)\cdot u_s^{f_2}(c))\mathrm{d}w \\
        & = u_s^{f_1}(c)\int_0^1 w\mathrm{d}\mathrm{Pr}[w_s^{f_1}\leq w] + u_s^{f_2}(c) \int_0^1 (1-w)\mathrm{d}\mathrm{Pr}[w_s^{f_1}\leq w] \\
        & = (1-A) u_s^{f_1}(c) + Au_s^{f_2}(c).
    \end{align*}
    On the other hand, we also have
    \[ \mathbb{E}[u_{s'}(c)] = \frac{u_{s'}^{f_1}(c)+u_{s'}^{f_2}(c)}{2} = (1-A) u_s^{f_1}(c) + Au_s^{f_2}(c), \]
    which is the same as $\mathbb{E}[u_s(c)]$.

    Secondly, for any two colleges $c_i$, $c_j\in M$, we have $u_{s'}^{f_1}(c_i) - u_{s'}^{f_1}(c_j) = 2(1-A)[u_{s}^{f_1}(c_i) - u_{s}^{f_1}(c_j)]$ and $u_{s'}^{f_2}(c_i) - u_{s'}^{f_2}(c_j) = 2A[u_{s}^{f_2}(c_i) - u_{s}^{f_2}(c_j)]$. Since $A = \int_0^1 \mathrm{Pr}[w_s^{f_1}\leq w] \mathrm{d}w = 1-\mathbb{E}[w_s^f]$ and $\mathrm{Pr}[w_s^{f_1}\leq \mathbb{E}[w_s^{f_1}]] = 1/2$, then $0<A<1$. Hence, for both $f\in \{f_1,f_2\}$, the difference between $u_{s'}^{f}(c_i)$ and $u_{s'}^{f}(c_j)$ always has the same sign with that between $u_{s}^{f}(c_i)$ and $u_{s}^{f}(c_j)$, which implies that the second condition in the definition is satisfied.
    Lastly, for continuous probability distribution, $\mathrm{Pr}[c_i\succeq_s^{w_s}c_j] \geq \mathrm{Pr}[c_j\succeq_s^{w_s}c_i]$ is the same to the condition $\mathrm{Pr}[c_i\succeq_s^{w_s}c_j] \geq 1/2$. By the fact proved in the second point, we know that for any college $c_i$, $c_j\in M$, $\mathrm{Pr}[c_i\succeq_s^{w_s}c_j]$ and $\mathrm{Pr}[c_i\succeq_{s'}^{w_{s'}}c_j]$ must have the same condition as in the formula given in Lemma~\ref{lem:probformula}. Then, the third condition must be satisfied when $\mathrm{Pr}[c_i\succeq_s^{w_s}c_j] = 1$. For the remaining two cases, by symmetry, we assume that $u_s^{f_1}(c_i) > u_{s}^{f_1}(c_j)$ and $u_s^{f_2}(c_i) < u_{s}^{f_2}(c_j)$ w.l.o.g., and we have $\mathrm{Pr}[c_i\succeq_s^{w_s}c_j] = \mathrm{Pr}[w_s\geq 1/(1+x)]$ and $\mathrm{Pr}[c_i\succeq_{s'}^{w_{s'}}c_j] = \mathrm{Pr}\left[w_{s'}\geq 1\left/\left(1+\frac{1-A}{A}x\right)\right]\right.$, where $x=\frac{\Delta_s^{f_1}(c_i,c_j)}{\Delta_s^{f_2}(c_i,c_j)}$. Then, $\mathrm{Pr}[c_i\succeq_s^{w_s}c_j] \geq 1/2$ is equivalent to $1/(1+x) \leq 1-A$, and further $(1-A)x\geq A$. By adding $A$ on both sides, we have $A+(1-A)x\geq 2A$, which is equivalent to $1\left/\left(1+\frac{1-A}{A}x\right)\right.\leq \frac{1}{2}$. Therefore, $\mathrm{Pr}[c_i\succeq_{s'}^{w_{s'}}c_j] = \mathrm{Pr}\left[w_{s'}\geq 1\left/\left(1+\frac{1-A}{A}x\right)\right]\right. \geq \frac{1}{2}$ holds equivalently.
\end{proof}

With this equivalence, the properties of a matching, such as social welfare, potential blocking pairs, and incentives for misreporting, remain unchanged. Moreover, the outcomes of methods based on these comparisons, such as HEUF and LOICV, are also preserved.

\subsection{Computational Issues for More Features}\label{sec:moretc}
In Proposition~\ref{prop:calProS2}, we only give a polynomial algorithm to compute $\mathsf{ProS}(\pi)$ for a certain matching $\pi$ when $|F| = 2$. Here, we supplement the computational issues when $|F|\geq 3$.

\begin{lemma}\label{lem:probformula3d}
    For any student $s\in N$ and any two colleges $c_1, c_2\in M$, denote $\delta_s^f(c_i,c_j) = u_s^f(c_i) - u_s^f(c_j)$ for any $f\in F$, a vector of dimension $|F|-1$ as $\vec{d}_{s}(c_i,c_j) = (\delta_s^{f_{|F|}}(c_i,c_j) - \delta_s^{f_{k}}(c_i,c_j))_{1\leq k<|F|}$, and a weight vector of the same dimension $\vec{w}_s = (w_s^f)_{1\leq f<|F|}$. We have $\mathrm{Pr}[c_i\succ_{s}^{w_s} c_j] = \mathrm{Pr}[\vec{d}_{s}(c_i,c_j)^{\mathsf{T}}\vec{w}_s < \delta_s^{f_{|F|}}(c_i,c_j)]$. 
\end{lemma}
\begin{proof}
    By definition, we can calculate that
    \begin{align*}
        & \mathrm{Pr}[c_i\succ_s^{w_s} c_j] \\
        & = \mathrm{Pr}\left[\sum_{k=1}^{|F|-1}w_s^{f_k}u_s^{f_k}(c_i) + \left(1 - \sum_{k=1}^{|F|-1}w_s^{f_k}\right)u_s^{f_{|F|}}(c_i)\geq \right. \\
        & \qquad \left. \sum_{k=1}^{|F|-1}w_s^{f_k}u_s^{f_k}(c_j) + \left(1 - \sum_{k=1}^{|F|-1}w_s^{f_k}\right)u_s^{f_{|F|}}(c_j)\right] \\
        & = \mathrm{Pr}\left[\sum_{k=1}^{|F|-1}w_s^{f_k}\left(\delta_s^{f_{|F|}}(c_i,c_j) - \delta_s^{f_k}(c_i,c_j)\right) \leq \delta_s^{f_{|F|}}(c_i,c_j)\right] \\
        & = \mathrm{Pr}[\vec{d}_{s}(c_i,c_j)^{\mathsf{T}}\vec{w}_s < \delta_s^{f_{|F|}}(c_i,c_j)].
    \end{align*}
    This concludes the proof.
\end{proof}

Notice that the formula in Lemma~\ref{lem:probformula3d} actually computes the cumulative probability over an $(|F|-1)$-dimensional half-space. Even with access to an oracle of c.d.f, the time complexity of computing such a probability depends on the form of the probability distribution. In general, calculating the volume of such an integral is \#P-hard with respect to the number of dimensions~\citep{khachiyan1993complexity}, which renders the computational intractability when $|F|$ is large. Consequently, calculating $\mathsf{ProS}$, the methods LOCV, LOICV, and HERF are intractable in the most general cases.

In practice, when eliciting probability distributions, it is often more realistic to consider discrete distributions with finite support, or to query the parameters of specific classes of distributions such as the uniform or normal distribution. In the former case, the computation is polynomial in the input size of the probability distribution, since we can enumerate all points located in the intersections of the target spaces. In the latter case, with specific distribution families and a relatively small number of features, we can still maintain a reasonable time complexity~\citep{preparata1979finding,sun1988general} or employ standard approximation algorithms~\citep{ridgway2016computation}. We supplement the NP-hardness of finding the highest $\mathsf{ProS}$ even with additional restrictions of probability distributions and utility functions as follows.

\begin{theorem}\label{thm:nphard3}
    It is NP-hard to find a matching with the highest $\mathsf{ProS}$ in a school choice with feature-based uncertainty even when all $\mu_s\in \mu$ is uniform distribution and for each $s\in N$, $f\in F$, $c_i,c_j\in M$ such that $u_s^f(c_i)\neq u_s^f(c_j)$, it has $|u_s^f(c_i) - u_s^f(c_j)| \geq \frac{1}{m}$.
\end{theorem}
\begin{proof}
    We prove the NP-hardness by reducing the problem of \textit{highest probability of stable marriage with compact indifference model} (HSMCI)~\cite{aziz2020stable}. An instance of HSMCI is a problem of one-to-one two sided matching where each agent could report a single weak preference list that allows for ties. Each complete strict linear order extension of a weak order is assumed to be equally likely. The target is to find a matching that has the highest probability of being a stable marriage. \citeauthor{aziz2020stable}~\cite{aziz2020stable} have shown that the HSMCI is NP-hard even if only one side of the market has uncertain agents. Then consider an instance of the HSMCI problem, with only man-side has uncertain agents. We can reduce it to our problem as follows.
    \begin{itemize}
        \item Correspond each man as a student, and each woman as a college;
        \item For each college, it has capacity of one, and the same preference with the corresponding woman;
        \item Let $l$ be the maximum length of all ties in the instance of HSMCI, and let the number of features be $|F| = l$;
        \item For each student $s$, her weight vector is uniformly distributed over all possibilities;
        \item For each student $s$, her utility functions is defined by the following process: (i) traverse the weak preference of the corresponding man from the least preferred one to the most preferred one, and let the sequence of the corresponding colleges is $c_{(1)}$, $c_{(2)}$, $\dots$, $c_{(m)}$, where $m$ is the number of women; (ii) for each $c_{(i)}$, if it is not involved in any tie, then let $u_s^f(c_{(i)}) = i/m$ for any $f\in F$; (iii) if from $c_{(i)}$ to $c_{(i+t)}$, $t\leq l$, it forms a tie and $c_{(i)}$ is the first college in the tie, then let $u_s^{f_{k+1}}(c_{(i+k)}) = (i+1)/m$ for all $0\leq k\leq t$, and any other utilities equal to $i/m$.
    \end{itemize}
    The overall reduction could be done in polynomial time, and all possible aggregated preferences are corresponds to an extension of the weak order in HSMCI with the same probability. Hence, the original instance of HSMCI and the reduced instance indeed share the same solution. Therefore, we can conclude that our problem is also in NP-hard even with the given constraints.
\end{proof}

Additionally, we have a failure of Lemma~\ref{lem:transitivity} for $|F|\geq 3$.

\begin{example}\label{ex:notrans}
    We illustrate a case to indicate that the transitivity characterized in Lemma~\ref{lem:transitivity} is not generally hold for instances with $|F|>2$. Consider a student $s\in N$ with:
    \begin{itemize}
        \item $u_s^{f_1}(c_1) = 0.65$, $u_s^{f_1}(c_2) = 0.91$, $u_s^{f_1}(c_3) = 0.10$;
        \item $u_s^{f_2}(c_1) = 0.90$, $u_s^{f_2}(c_2) = 0.05$, $u_s^{f_2}(c_3) = 1.00$;
        \item $u_s^{f_3}(c_1) = 0.21$, $u_s^{f_3}(c_2) = 0.31$, $u_s^{f_3}(c_3) = 0.70$.
    \end{itemize}
    Then according to the formula given in Lemma~\ref{lem:probformula3d}, $\mathrm{Pr}[c_1\succ_{s}^{w_s} c_2]$ equals to the probability of $(w_s^{f_1}, w_s^{f_2})$ being located right to the blue line in Figure~\ref{fig:nontrans}, $\mathrm{Pr}[c_2\succ_{s}^{w_s} c_3]$ equals to the probability of $(w_s^{f_1}, w_s^{f_2})$ being located above the green line, and $\mathrm{Pr}[c_1\succ_{s}^{w_s} c_3]$ equals to the probability of $(w_s^{f_1}, w_s^{f_2})$ being located right to the yellow line (the azure area). The probabilities of being located in each area are also labeled in the Figure. Finally, we have $\mathrm{Pr}[c_1\succ_{s}^{w_s} c_2] = \frac{1}{2}+\epsilon$ and $\mathrm{Pr}[c_2\succ_{s}^{w_s} c_3] = \frac{1}{2}+\epsilon$, but $\mathrm{Pr}[c_1\succ_{s}^{w_s} c_3] = 4\epsilon < \frac{1}{2}$.
\end{example}

\begin{figure}[t]
    \centering
    \includegraphics[width=0.3\linewidth]{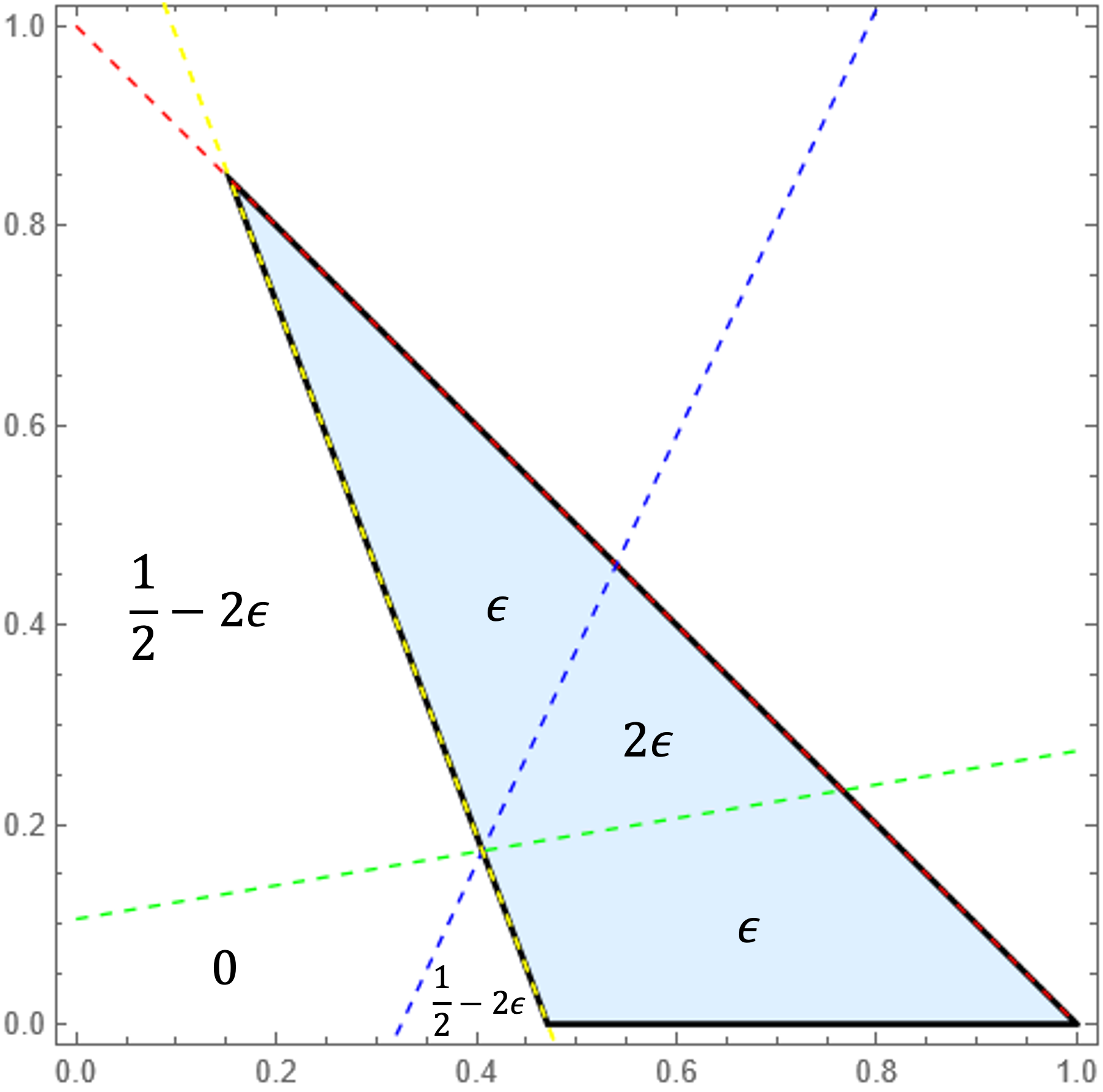}
    \caption{An example of the probability distribution of $\vec{w}_s$.}
    \label{fig:nontrans}
\end{figure}

Unfortunately, it also suggests an impossibility of IC-R algorithms when $|F|\geq 3$ with a similar proof as in Proposition~\ref{prop:noICA}.

\begin{corollary}\label{coro:impossibleICR}
    No matching algorithm can satisfy the property of IC-R unless it always outputs the same result when $|F|\geq 3$.
\end{corollary}
\begin{proof}
    Suppose a matching algorithm is IC-R when $|F|\geq 3$. Then, consider a student $s$ who has utility function and probability distribution all the same as the one in Example~\ref{ex:notrans}. Now, we have $\mathrm{Pr}[c_i\succ_s^{w_s}c_j] > 1/2$ for any $c_i\neq c_j$. Suppose the algorithm in this case matches $s$ to a college $c$ ($c$ could be $\mathsf{null}$ if she is not assigned to any college). Then, for any other possible utility functions $u_s'$ and probability distribution $\mu_s'$, the algorithm must still assign $s$ to the same $s$; otherwise, $s$ can misreport her information between two cases to obtain another college that is more likely to be better. Hence, the algorithm could only output the same result for any input.
\end{proof}
\section{Conclusion and Future Work}\label{sec:conclusion}
In this paper, we study a feature-based uncertainty model for the school choice problem and analyze a class of algorithms that achieve different levels of incentive compatibility (IC) and ProS guarantees. Several important directions remain for future research.

Our model assumes that a probability distribution over weight vectors is available. In practice, obtaining an arbitrary full-support distribution may be difficult, but our results do not rely on exact specifications—reasonable approximations of students’ valuation uncertainty are sufficient. Such approximations are often obtainable in real systems through application data, surveys, or preference elicitation tools that evaluate programs across multiple attributes. Empirical methods such as discrete choice models and conjoint analysis can also be used to estimate preference weights from data~\citep{soutar2002students}. Thus, while exact distributions may be unrealistic, structured or discretized approximations are practical and compatible with our framework. Performance may be further improved if signaling mechanisms help shape uncertainty into more structured forms.

Another realistic factor is that uncertainty may be endogenous: colleges or policymakers can influence it by signaling or information design. Although our worst-case analysis ensures robustness under such effects, it remains an open question whether improved outcomes can be achieved through explicit information design.

Finally, understanding how different uncertainty structures systematically affect matching outcomes and relate to axiomatic properties is a promising direction. Our examples suggest several patterns—for instance, more symmetric feature-weight distributions may induce acyclic comparisons, while similarity structures across features or colleges may favor different algorithms (e.g., high similarity potentially favoring LOICV). Providing a quantitative characterization of these effects is an important topic for future work.

\vspace{3mm}
\noindent \textbf{Acknowledgments.} This work is supported by JST ERATO Grant Number JPMJER2301.

\bibliographystyle{plainnat}
\bibliography{mybibfile}

@article{aziz2020stable,
  title={Stable matching with uncertain linear preferences},
  author={Aziz, Haris and Bir{\'o}, P{\'e}ter and Gaspers, Serge and de Haan, Ronald and Mattei, Nicholas and Rastegari, Baharak},
  journal={Algorithmica},
  volume={82},
  number={5},
  pages={1410--1433},
  year={2020},
  publisher={Springer}
}

@article{manlove2002hard,
  title={Hard variants of stable marriage},
  author={Manlove, David F and Irving, Robert W and Iwama, Kazuo and Miyazaki, Shuichi and Morita, Yasufumi},
  journal={Theoretical Computer Science},
  volume={276},
  number={1-2},
  pages={261--279},
  year={2002},
  publisher={Elsevier}
}

@article{halldorsson2007improved,
  title={Improved approximation results for the stable marriage problem},
  author={Halld{\'o}rsson, Magn{\'u}s M and Iwama, Kazuo and Miyazaki, Shuichi and Yanagisawa, Hiroki},
  journal={ACM Transactions on Algorithms (TALG)},
  volume={3},
  number={3},
  pages={30--es},
  year={2007},
  publisher={ACM New York, NY, USA}
}

@inproceedings{chu2024stable,
  title={Stable matching with approval preferences under partial information},
  author={Chu, Yaqin and Luo, Junjie and Zheng, Tianyang},
  booktitle={International Conference on Algorithmic Aspects in Information and Management},
  pages={64--75},
  year={2024},
  organization={Springer}
}

@article{soutar2002students,
  title={Students’ preferences for university: A conjoint analysis},
  author={Soutar, Geoffrey N and Turner, Julia P},
  journal={International journal of educational management},
  volume={16},
  number={1},
  pages={40--45},
  year={2002},
  publisher={MCB UP Ltd}
}

@article{haeringer2009constrained,
  title={Constrained school choice},
  author={Haeringer, Guillaume and Klijn, Flip},
  journal={Journal of Economic theory},
  volume={144},
  number={5},
  pages={1921--1947},
  year={2009},
  publisher={Elsevier}
}

@article{kimura2025multi,
  title={Multi-stage generalized deferred acceptance mechanism: Strategyproof mechanism for handling general hereditary constraints},
  author={Kimura, Kei and Liu, Kweiguu and Sun, Zhaohong and Yahiro, Kentaro and Yokoo, Makoto},
  journal={Autonomous Agents and Multi-Agent Systems},
  volume={39},
  number={2},
  pages={1--25},
  year={2025},
  publisher={Springer}
}

@inproceedings{alimudin2021study,
  title={A Study of the Random Order Mechanism for Uncertain Preferences in the Stable Marriage Problem},
  author={Alimudin, Akhmad and Ishida, Yoshiteru},
  booktitle={2021 8th International Conference on Advanced Informatics: Concepts, Theory and Applications (ICAICTA)},
  pages={1--6},
  year={2021},
  organization={IEEE}
}

@article{aziz2022stable,
  title={Stable matching with uncertain pairwise preferences},
  author={Aziz, Haris and Bir{\'o}, P{\'e}ter and Fleiner, Tam{\'a}s and Gaspers, Serge and De Haan, Ronald and Mattei, Nicholas and Rastegari, Baharak},
  journal={Theoretical Computer Science},
  volume={909},
  pages={1--11},
  year={2022},
  publisher={Elsevier}
}

@article{dvir2020modelling,
  title={Modelling the expected probability of correct assignment under uncertainty},
  author={Dvir, Tom and Peres, Renana and Rudnick, Ze{\'e}v},
  journal={Scientific Reports},
  volume={10},
  number={1},
  pages={15080},
  year={2020},
  publisher={Nature Publishing Group UK London}
}

@article{aziz2024approval,
  title={Approval-based committee voting under uncertainty},
  author={Aziz, Haris and Kagita, Venkateswara Rao and Rastegari, Baharak and Suzuki, Mashbat},
  journal={arXiv preprint arXiv:2407.19391},
  year={2024}
}

@inproceedings{aziz2024envy,
  title={Envy-free house allocation under uncertain preferences},
  author={Aziz, Haris and Iliffe, Isaiah and Li, Bo and Ritossa, Angus and Sun, Ankang and Suzuki, Mashbat},
  booktitle={Proceedings of the AAAI Conference on Artificial Intelligence},
  volume={38},
  pages={9477--9484},
  year={2024}
}

@article{gale2013college,
  title={College admissions and the stability of marriage},
  author={Gale, David and Shapley, Lloyd S},
  journal={The American Mathematical Monthly},
  volume={120},
  number={5},
  pages={386--391},
  year={2013},
  publisher={Taylor \& Francis}
}

@article{owen2020student,
  title={Student preferences for college and career information},
  author={Owen, Laura and Poynton, Timothy A and Moore, Raeal},
  journal={Journal of college access},
  volume={5},
  number={1},
  pages={7},
  year={2020}
}

@article{bell2009college,
  title={College knowledge of 9th and 11th grade students: Variation by school and state context},
  author={Bell, Angela D and Rowan-Kenyon, Heather T and Perna, Laura W},
  journal={The Journal of Higher Education},
  volume={80},
  number={6},
  pages={663--685},
  year={2009},
  publisher={Taylor \& Francis}
}

@article{hoxby2015high,
  title={What high-achieving low-income students know about college},
  author={Hoxby, Caroline M and Turner, Sarah},
  journal={American Economic Review},
  volume={105},
  number={5},
  pages={514--517},
  year={2015},
  publisher={American Economic Association 2014 Broadway, Suite 305, Nashville, TN 37203}
}

@misc{hastings2007preferences,
  title={Preferences, information, and parental choice behavior in public school choice},
  author={Hastings, Justine S and Van Weelden, Richard and Weinstein, Jeffrey M},
  year={2007},
  publisher={National Bureau of Economic Research Cambridge, Mass., USA}
}

@article{roth2008deferred,
  title={Deferred acceptance algorithms: History, theory, practice, and open questions},
  author={Roth, Alvin E},
  journal={international Journal of game Theory},
  volume={36},
  number={3-4},
  pages={537--569},
  year={2008},
  publisher={Springer}
}

@article{abdulkadirouglu2003school,
  title={School choice: A mechanism design approach},
  author={Abdulkadiro{\u{g}}lu, Atila and S{\"o}nmez, Tayfun},
  journal={American economic review},
  volume={93},
  number={3},
  pages={729--747},
  year={2003},
  publisher={American Economic Association}
}

@inproceedings{aziz2019matching,
  title={From Matching with Diversity Constraints to Matching with Regional Quotas},
  author={Aziz, Haris and Gaspers, Serge and Sun, Zhaohong and Walsh, Toby},
  booktitle={Proceedings of the 18th International Conference on Autonomous Agents and MultiAgent Systems},
  pages={377--385},
  year={2019}
}

@inproceedings{takeshima2025new,
  title={A New Relaxation of Fairness in Two-Sided Matching Respecting Acquaintance Relationships},
  author={Takeshima, Ryota and Kimura, Kei and Kuroki, Ayumu and Wakasugi, Temma and Yokoo, Makoto},
  booktitle={28th European Conference on Artificial Intelligence, ECAI 2025, including 14th Conference on Prestigious Applications of Intelligent Systems, PAIS 2025},
  pages={3727--3734},
  year={2025},
  organization={IOS Press BV}
}

@inproceedings{cho2022two,
  title={Two-Sided Matching over Social Networks},
  author={Cho, Sung-Ho and Todo, Taiki and Yokoo, Makoto},
  booktitle={Proceedings of the Thirty-First International Joint Conference on Artificial Intelligence},
  year={2022},
  organization={International Joint Conferences on Artificial Intelligence Organization}
}

@incollection{khachiyan1993complexity,
  title={Complexity of polytope volume computation},
  author={Khachiyan, Leonid},
  booktitle={New trends in discrete and computational geometry},
  pages={91--101},
  year={1993},
  publisher={Springer}
}

@article{preparata1979finding,
  title={Finding the intersection of n half-spaces in time O (n log n)},
  author={Preparata, Franco P. and Muller, David E.},
  journal={Theoretical Computer Science},
  volume={8},
  number={1},
  pages={45--55},
  year={1979},
  publisher={Elsevier}
}

@article{ridgway2016computation,
  title={Computation of Gaussian orthant probabilities in high dimension},
  author={Ridgway, James},
  journal={Statistics and computing},
  volume={26},
  number={4},
  pages={899--916},
  year={2016},
  publisher={Springer}
}

@article{sun1988general,
  title={A general reduction method for n--variate normal orthant probability},
  author={Sun, Hong--Jie},
  journal={Communications in Statistics-Theory and Methods},
  volume={17},
  number={11},
  pages={3913--3921},
  year={1988},
  publisher={Taylor \& Francis}
}

@techreport{abdulkadiroglu2018impact,
  title={Impact evaluation in matching markets with general tie-breaking},
  author={Abdulkadiro{\u{g}}lu, Atila and Angrist, Joshua D and Narita, Yusuke and Pathak, Parag A},
  year={2018},
  institution={National Bureau of Economic Research}
}

@article{abdulkadi2022breaking,
  title={Breaking ties: Regression discontinuity design meets market design},
  author={Abdulkadiro{\u{g}}lu, Atila and Angrist, Joshua D and Narita, Yusuke and Pathak, Parag},
  journal={Econometrica},
  volume={90},
  number={1},
  pages={117--151},
  year={2022},
  publisher={Wiley Online Library}
}

\end{document}